\documentclass{article}

\usepackage{a4wide}
\usepackage[english]{babel}
\usepackage{latexsym}
\usepackage{amssymb}
\usepackage{amsmath}
\usepackage{amsthm}
\usepackage{epstopdf}
\usepackage{graphicx}
\usepackage{graphics}
\usepackage{cite}    
\usepackage{ifthen}
\usepackage{xspace}
\usepackage{paralist}
\usepackage{multirow}
\usepackage{shuffle}
\usepackage[colorlinks,final,citecolor=blue]{hyperref}

\usepackage{marginnote}
\usepackage{verbatim}  

\usepackage[usenames,dvipsnames]{xcolor}
\usepackage{tikz}
\usetikzlibrary{%
  arrows,%
  positioning,%
  decorations.pathmorphing,%
  decorations.pathreplacing,%
  automata,%
}

\newenvironment{transducer}[1][]
	{\begin{center}
	\begin{tikzpicture}[font=\footnotesize,shorten >=1pt,node distance=3cm,
	on grid,>=stealth',initial text=,every node/.style={align=center},
	every state/.style={inner sep=1pt},
	every path/.style={->,bend angle=20},#1]}
	{\end{tikzpicture}\end{center}}

\newcommand\pssi{\par\smallskip\indent}
\newcommand\pmsi{\par\medskip\indent}

\newcommand\pssn{\par\smallskip\noindent}
\newcommand\pmsn{\par\medskip\noindent}
\newcommand\pbsn{\par\bigskip\noindent}
\newcommand\pnsi{\par\indent}
\newcommand\pnsn{\par\noindent}

\newcommand\emdef[1]{\emph{#1}}  
\newcommand\emshort[1]{\emph{\textbf{#1}}}

\theoremstyle{plain}
\newtheorem{theorem}{Theorem}[]
\newtheorem{proposition}[theorem]{Proposition}
\newtheorem{lemma}[theorem]{Lemma}

\theoremstyle{definition}
\newtheorem{definition}[theorem]{Definition}
\newtheorem{example}[theorem]{Example}

\newenvironment{EX}{\begin{example}}{$\>\>\Box$\end{example}}

\theoremstyle{remark}
\newtheorem{remark}[theorem]{Remark}

\newcommand{\ttp}{\texttt{p}}
\newcommand{\ttps}{\texttt{p}\;}
\newcommand{\ttq}{\texttt{q}}
\newcommand{\ttqs}{\texttt{q}\;}

\newcommand{\cP}{\ensuremath{\mathcal{P}}\xspace}
\newcommand{\cQ}{\ensuremath{\mathcal{Q}}\xspace}

\renewcommand{\phi}{\varphi}
\renewcommand{\epsilon}{\varepsilon}

\newcommand\true{\texttt{True }}

\newcommand\falsen{\texttt{False}}
\newcommand\none{\texttt{None }}
\newcommand\nonen{\texttt{None}}

\newcommand{\sse}{\subseteq}

\newcommand{\es}{\emptyset}

\newcommand\card[1]{|#1|}
\newcommand\pset[1]{2^{#1}}  

\newcommand{\N}{\mathbb{N}}
\newcommand{\Z}{\mathbb{Z}}

\newcommand{\R}{\mathbb{R}}

\newcommand{\e}{\epsilon}
\newcommand{\ew}{\e}      
\newcommand\al{\Sigma}        
\newcommand\alD{\Delta}        

\newcommand\pty{\cP}  
\newcommand\iatp{\pty^{\mathrm al}}
\newcommand\edp{\pty^{\mathrm ed}}
\newcommand\ecp{\pty^{\mathrm ec}}

\newcommand{\ree}{\bar{e}}  
\newcommand\aut{\mathbf{a}}   
\newcommand\autb{\mathbf{b}}   
\newcommand{\aai}{\texttt{a }}   
\newcommand{\abi}{\texttt{b }}   
\newcommand{\aain}{\texttt{a}}   
\newcommand\lang{\ensuremath\mathrm{L}\xspace} 
\newcommand\tr{\mathbf{t}}    
\newcommand{\trt}{\mathbf{t}}   
\newcommand\trs{\mathbf{s}}      
\newcommand\tru{\mathbf{u}}      
\newcommand{\tti}{\texttt{t }}   
\newcommand{\tsi}{\texttt{s }}   
\newcommand{\ttin}{\texttt{t}}   
\newcommand{\tsin}{\texttt{s}}   
\newcommand{\ppi}{\texttt{p }}   
\newcommand{\ppin}{\texttt{p}}   
\newcommand\rel{\ensuremath\mathrm{R}\xspace} 
\newcommand\sz[1]{|#1|}       
\newcommand\compose{\circ}  

\newcommand\po{\le}             
\newcommand\sop{\&}           
\newcommand\rep[1]{[#1]}      
\newcommand\gene[1]{\langle#1\rangle}   
\newcommand{\psetq}{\cQ}  
\newcommand\szg[1]{\|#1\|}
\newcommand\diag{\mathrm{diag}}

\begin{document}

\begin{center}
\textbf{\Large Symbolic Manipulation of Code Properties
}
\pbsn
{\large Stavros Konstantinidis$^{1}$, Casey Meijer$^{1}$,
Nelma Moreira$^{2}$, Rog{\'e}rio Reis$^{2}$}
\end{center}

\pmsn
$^{1}$ Department of Mathematics and Computing Science,
Saint Mary's University, Halifax, Nova Scotia, Canada,
\texttt{s.konstantinidis@smu.ca, dylanyoungmeijer@gmail.com}
\pmsn
$^{2}$  CMUP \& DCC, Faculdade de Ci{\^e}ncias da Universidade do Porto,
Rua do Campo Alegre, 4169–007 Porto Portugal
\texttt{\{nam,rvr\}@dcc.fc.up.pt}

\pbsn
\textbf{Abstract.}
%

The FAdo system is a symbolic manipulator of formal languages objects, implemented in Python. In this work, we extend its capabilities by implementing methods to manipulate transducers and we go one level higher than existing formal language systems and implement methods to manipulate objects representing
classes of independent languages (widely known as code properties).
Our methods allow users to define their own code properties and combine them
between themselves or with fixed properties such as prefix codes, suffix codes,
error detecting codes, etc. The satisfaction and maximality decision questions are
solvable for any of the definable properties. The new online system LaSer allows to query about code properties and obtain the answer in a batch mode.
Our work is founded on independence theory as well as the theory 
of rational relations and  transducers and contributes 
with improveded  algorithms on these objects.

\pbsn

\section{Introduction}\label{sec:intro}
Several programming platforms are nowadays available, providing methods to transform and manipulate various formal
language objects: Grail/Grail+\cite{RayWood:1994,grail}, Vaucanson 2~\cite{Vauc,DDLS:2013}, FAdo \cite{FAdo,AAAMR:2009},
OpenFST \cite{OpenFST}, JFLAP \cite{OpenFST}.
Some of these systems allow one to manipulate such objects within
simple script environments. Grail for example,
one of the oldest systems, provides a set of  filters
manipulating automata and regular expressions on a UNIX command shell. Similarly, FAdo provides a set of methods manipulating such
objects on a Python shell~\cite{Python}.

Software environments for symbolic manipulation of formal languages are widely recognized as important tools for theoretical and practical research. 
 They allow easy   prototyping of new algorithms, testing algorithm performance with large datasets, corroborate or disprove descriptional complexity bounds for manipulations of formal systems representations, etc. A typical example is, for a given operation on regular languages, to find an upper bound for the number states of a minimal deterministic finite automata (DFA) for the language that results from the operation, as a function of the number of states of the minimal DFAs of the operands. Due to the combinatorial nature of formal languages representations this kind of calculations are almost impossible without computational add.

In this work, we extend  the capabilities of FAdo and
LaSer~\cite{DudKon:2012,Laser} by implementing
transducer methods and by going to the higher level of implementing
objects representing classes of independent formal languages, also known as
code properties.
More specifically, the contributions of the present paper are as follows.
\begin{enumerate}
\item
Implementation of transducer objects and state of the art transducer methods. FAdo is a regularly maintained and user friendly
Python package that was lacking transducer objects.
Now available are general transducers
as well as transducers in standard and normal forms.
Some important methods that have
been implemented are various product constructions between transducers and between
transducers and automata, as well as    a  
transducer functionality test.
\item
Definitions of objects representing code properties and methods for their manipulation,
which to our knowledge is
a new development in software related to formal language objects. In simple mathematical terms, a code property is
a class of languages that is closed for maximal languages.
In addition to some fixed known code properties (such as prefix code, suffix code, hypercode), these methods can be used to construct new
code properties, including various error-detecting properties, which are specified either via a trajectory regular expression~\cite{Dom:2004} or via a transducer description~\cite{DudKon:2012}. Moreover, our methods can
be used to combine any defined properties and produce a new property that is the conjunction of the defined properties.
\item
Enhancement and implementation of state of the art decision algorithms for code properties
of  regular languages. In particular, many such algorithms have been implemented and enhanced so
as to provide witnesses (counterexamples) in case of a negative answer,
for example, when the given regular language does not satisfy the property,
or is not maximal with respect to the property. To our knowledge such
implementations are not openly available. In particular, the
satisfaction of the classic property of unique decipherability, or decodability, is implemented for any
given NFA based on the algorithm in~\cite{Head:Weber:decision} as well as
the satisfaction of various error-detecting properties~\cite{Kon:2002}.
\item
A mathematical definition of what it means to simulate (and hence implement) a
hierarchy of properties and the proof that there is no complete simulation
of the set of error-detecting properties.
\item
Generation of executable Python code based on the requested question about a given code property.
This is mostly of use in the online LaSer~\cite{Laser}, which receives client requests and attempts to compute answers. However, as
the algorithm required to compute an answer can take a long time, LaSer provides the option
to compute and return a self-contained executable Python program that can be executed at the client's machine and return the required answer.
\item
All the above classes and methods are open source (GPL)
: available for anyone to copy, use and modify to suit their own application.
\end{enumerate}
Our work is founded on independence theory \cite{Shyr:Thierrin:relations,JuKo:handbook}
as well as the theory 
of rational relations and  transducers \cite{Be:1979,Sak:2009}.
We present our algorithmic enhancements 
in a detailed mathematical manner with \textit{two aims in mind:}
first to establish the correctness
of the enhanced algorithms and, second to allow interested readers to
obtain a deeper understanding of these algorithms, which could
potentially lead to further developments.

\pssn
The paper is organized as follows.
\begin{description}
  \item[Section 2] contains some basic terminology and background about various formal language concepts as well as a few examples of manipulating FAdo automata in Python.
  \item[Section 3] describes our implementation of transducer object classes and a few basic
    methods involving product constructions and rational operations.
  \item[Section 4]  describes the decision algorithm for transducer functionality, and then our enhancement  so as to provide witnesses when the transducer in question is not functional.
  \item[Section 5] describes our implementation of code property objects and basic methods for their manipulation. Moreover, a mathematical approach to defining syntactic simulations of
  infinite sets of properties is presented, explaining that a linear simulation of all error-detecting
  properties exists, but no complete simulation of these properties is possible.
  \item[Section 6] continues on code property methods by describing our implementation of the satisfaction and maximality methods.
  Again, we describe our enhancements so as to provide witnesses when the answer to the satisfaction/maximality question is negative.
  \item[Section 7] describes the implementation of the unique decodability (or decipherability) property
      and its satisfaction and maximality algorithms. This property is presented separately, as it is a classic one that cannot
      be defined within the methods of transducer properties.
  \item[Section 8] describes the new version of LaSer~\cite{Laser} including
   the capability to generate executable programming code in response to a
   client's request for the satisfaction or maximality of a given code property.
  \item[Section 9] contains a few concluding remarks including directions for future research.
\end{description}

\section{Terminology and Background}\label{sec:terminology}

\emshort{Sets, alphabets, words, languages.}
We write $\N,\N_0,\Z,\R$ for the sets of natural numbers (not including 0), non-negative
integers, integers, and real numbers, respectively.
If $S$ is a set, then $\card S$ denotes the cardinality of $S$,
and $\pset S$ denotes the set of all subsets of $S$.
An \emdef{alphabet} is a finite nonempty set of symbols. In this paper,
we write $\al,\alD$ for any arbitrary alphabets.
If $q\in\N$, then
$\al_q$ denotes the alphabet $\{0,1,\ldots,q-1\}$.
The set of all words, or strings, over an alphabet $\al$ is written as
$\al^*$, which includes the \emdef{empty word} $\ew$.
A \emdef{language} (over $\al$) is any set of words.
In the rest of this paragraph, we use the following arbitrary object names: $i,j$ for nonnegative integers, $K,L$ for languages and $u,v,w,x,y$ for words.
If $w\in L$ then we say that $w$ is an \emdef{$L$-word}.
When there is no risk of confusion, we write a singleton language $\{w\}$
simply as $w$. For example,
$L\cup w$ and $v\cup w$  mean $L\cup\{w\}$ and $\{v\}\cup\{w\}$, respectively.
We use standard operations and notation on words and languages
\cite{HopcroftUllman,Wood:theory:of:comput,MaSa:handbook,FLHandbookI}.
For example, $uv$, $w^i$, $KL$, $L^i$, $L^*$, $L^+$ represent respectively, the concatenation of $u$ and $v$, the word consisting of $i$ copies of $w$, the concatenation of $K$ and $L$, the language consisting of all words obtained by
concatenating any $i$ $L$-words, the Kleene star of $L$, and $L^+=L^*\setminus\ew$. If $w$ is of the form $uv$ then $u$ is a
\emdef{prefix} and
$v$ is a \emdef{suffix} of $w$. If $w$ is of the form $uxv$ then $x$ is an
\emdef{infix} of $w$. If $u\not=w$ then $u$ is called a \emdef{proper prefix} of $w$---the definitions of proper suffix and proper infix are similar.
\pmsn
\emshort{Codes, properties, independent languages, maximality.}
A \emdef{property} (over $\al$) is any set $\pty$ of languages.
If $L$ is in $\pty$ then we say that $L$ \emdef{satisfies} $\pty$.
Let $\aleph_0$ denote the cardinality of $\N$.
A \emdef{code property}, or \emdef{independence}, \cite{JuKo:handbook},
is a property $\pty$ for which there is $n\in\N\cup\{\aleph_0\}$
such that
\[
L\in\pty, \quad\hbox{if and only if} \quad
L'\in\pty,\hbox{ for all $L'\sse L$ with $0<\card{L'}<n$,}
\]
that is, $L$ satisfies the property exactly when all nonempty subsets of $L$ with
less than $n$ elements satisfy the property. 
In this case, we also say that $\pty$ is an $n$-\emdef{independence}.
In the rest of the paper
we only consider properties $\pty$  that are code properties.
A language $L\in\pty$  is called \emdef{$\pty$-maximal}, or a
maximal $\pty$ code, if $L\cup w\notin\pty$ for any word $w\notin L$.
From \cite{JuKo:handbook} we have that every $L$ satisfying $\pty$ is
included in a maximal $\pty$ code. To our knowledge, all known code related properties in the
literature \cite{Ham:1950,Shyr:book,JuKo:handbook,SSYu:book,BePeRe:2009,Dom:2004,DudKon:2012,PAF:2013}
are code properties as defined above. For example, consider the `prefix code' property: $L$ is a \emdef{prefix code} if no word in $L$ is a proper prefix of a word in $L$.
This is a code property with $n=3$. Indeed, let $\pty'$ be the set of all singleton languages union
with the set of all languages $\{u,v\}$ such that $u$ is not a proper prefix of $v$ and vice versa. Then, every element in $\pty'$ is a prefix code. Moreover any $L$ is a prefix code if and only if any nonempty subset $L'$ of $L$ with less than three elements is in $\pty'$. As we shall see further below the focus of this work is on 3-independence
properties that can also be viewed as independent with respect to a binary relation in the
sense of \cite{Shyr:Thierrin:relations}.
\pmsn
\emshort{Automata and regular languages \cite{Yu:handbook,Sak:2009}.}
A  nondeterministic finite automaton with empty transitions, for short \emdef{automaton} or \emdef{$\ew$-NFA}, is a quintuple
$$\aut=(Q,\al, T,I, F)$$
such that $Q$ is the set of states, $\al$ is an alphabet, $I,F\subseteq Q$
are the sets of start (or initial) states and final states, respectively, and $T\subseteq Q\times(\al\cup\ew)\times Q$ is the finite set of \emdef{transitions}. Let $(p,x,q)$ be a transition of $\aut$. Then  $x$ is called the \emdef{label} of the transition, and we say that $p$ has an \emdef{outgoing} transition (with label $x$). A \emdef{path} 
 of $\aut$ is a finite sequence $(p_0,x_1,p_1,\ldots,x_\ell,p_\ell)$,
for some nonnegative integer $\ell$, such that each triple $(p_{i-1},x_i,p_i)$ is a transition of $\aut$. The word $x_1\cdots x_\ell$ is called the \emdef{label} of the path. The path is called \emdef{accepting} if
$p_0$ is a start state and $p_\ell$ is a final state. The \emdef{language accepted} by $\aut$, denoted as $\lang(\aut)$, is the set of labels of all the accepting paths of $\aut$. The $\ew$-NFA $\aut$ is called
\emdef{trim}, if every state appears in some accepting path of $\aut$.
The automaton $\aut$ is called an \emdef{NFA}, if no transition label is empty, that is, $T\subseteq Q\times\al\times Q$. A deterministic finite automaton, or \emdef{DFA} for short, is a special type of NFA where $I$ is a singleton set and there is no state $p$ having two outgoing transitions with equal labels.
The \emdef{size} $\sz{\aut}$ of the automaton $\aut$ is
$\card{Q}+\card{T}$. The automaton $\aut^{\ew}$ results when we add $\ew$-loops in $\aut$, that is, transitions $(p,\ew,p)$ for all states $p\in Q$.
\pmsn
\emshort{Transducers and (word) relations \cite{Be:1979,Yu:handbook,Sak:2009}.}
A (word) \emdef{relation} over $\al$ and $\alD$ is a subset
of $\al^*\times\alD^*$, that is, a set of pairs $(x,y)$ of words
over the two alphabets (respectively). The \emdef{inverse} of a relation $\rho$,
denoted as $\rho^{-1}$ is the relation $\{(y,x)\mid (x,y)\in \rho\}$.
A (finite) \emdef{transducer} is a sextuple
$$\tr=(Q, \al, \alD, T, I, F)$$
such that $Q,I,F$ are exactly the same as those in $\ew$-NFAs, $\al$ is now called the \
\emdef{input} alphabet, $\alD$ is the \emdef{output} alphabet, and $T\subseteq Q\times\al^*\times\alD^*\times Q$ is the finite set of transitions. We write $(p,x/y,q)$ for a transition -- the \emdef{label} here is $(x/y)$, with $x$ being the input and $y$ being the output label. The concepts of path, accepting path, and trim transducer are similar to those in $\ew$-NFAs. In particular the \emdef{label} of a path $(p_0,x_1/y_1,p_1,\ldots,x_\ell/y_\ell,p_\ell)$ is the pair $(x_1\cdots x_\ell,
y_1\cdots y_\ell)$ consisting of the input and output labels in the path. The \emdef{relation realized}
by the transducer $\tr$, denoted as $\rel(\tr)$, is the set of labels in all the accepting paths of $\tr$.
We write $\tr(x)$ for the set of \emdef{possible outputs of} $\tr$ on input $x$, that is,  $y\in\tr(x)$ iff $(x,y)\in \rel(\tr)$.
The \emdef{domain} of $\tr$ is the set of all words $w$ such that $\tr(w)\not=\emptyset$.
The \emdef{inverse} of a transducer $\tr$, denoted as $\tr^{-1}$, is the transducer that results from $\tr$ by simply switching the input and output alphabets of $\tr$ and also switching the input and output parts of the
labels in the transitions of $\tr$. It follows that $\tr^{-1}$ realizes the inverse of
the relation realized by $\tr$.
The transducer $\tr$ is said to be in \emdef{standard form}, if each transition $(p,x/y,q)$ is such that $x\in(\al\cup\ew)$ and $y\in(\alD\cup\ew)$.
It is in \emdef{normal form} if it is in standard form and exactly one
of $x$ and $y$ is equal to $\ew$.
We note that every transducer is effectively equivalent to one (realizing the same relation, that is) in standard form and one in normal form. 
As in the case of automata, the transducer $\trt^{\ew}$ results when we add $\ew$-loops in $\trt$, that is, transitions $(p,\ew/\ew,p)$ for all states $p\in Q$.
The \emdef{size} of a transition $(p,x/y,q)$ is the number $1+|x|+|y|$. The size $\sz{\tr}$ of the transducer $\tr$ is the sum of the
number of states and sizes of transitions in $T$. If $\trs$ and $\tr$ are
transducers, then there is a transducer $\trs\lor\tr$ of
size $O(|\trs|+|\tr|)$ realizing $\rel(\trs)\cup\rel(\tr)$.
\pmsn
\emshort{Automata and finite languages in FAdo \cite{FAdo}.}
The modules \texttt{fa} for automata, \texttt{fl} for finite languages, and \texttt{fio} for input/output of formal language
 objects, can be imported in a standard Python manner as follows.
\begin{verbatim}
import FAdo.fl as fl
import FAdo.fio as fio
from FAdo.fa import *    # import all fa methods for readability
\end{verbatim}
The FAdo object classes \texttt{FL}, \texttt{DFA} and \texttt{NFA}
manipulate finite languages, DFAs and $\ew$-NFAs, respectively.
\begin{EX}
The following code builds a finite language object \verb1L1 from a list of strings,
and then builds an NFA object \aai accepting the language
$\{a, ab, aab\}$.
\begin{verbatim}
lst = ['a', 'ab', 'aab']
L = fl.FL(lst)
a = L.toNFA()
\end{verbatim}
The second line uses the class name \verb1FL1 to create the finite language \verb1L1 from the given list of strings.
The last line returns an NFA accepting the language \verb1L1.
\end{EX}
\pbsn
The following code reads an automaton, or transducer, \aai from a file, and then writes
 it into a file.
\begin{verbatim}
a = fio.readOneFromFile(filename1)
 ...
fio.saveToFile(filename2, a)
\end{verbatim}
These methods assume that automata and transducers are written into those files in
FAdo format---see examples below and further ones
in \cite{FAdo,Laser}.
We can also read an automaton or transducer from a string \verb1s1 using the function \verb1readOneFromString(s)1.
\begin{EX}\label{exNFA1}
The following code defines a string  that contains an automaton accepting $a^*b$ and then uses that string to define an \verb1NFA1 object.
\begin{verbatim}
    st = '@NFA 1 * 0\n0 a 0\n0 b 1\n'
    a = fio.readOneFromString(st)
\end{verbatim}
As usual, the pattern \verb1\n1 denotes the \emdef{end of line character}, so the string \verb1st1 consists of three lines:
the first indicates the type of object followed by the final states (in this case 1) and the start states after \verb1*1 (in this case 0); the second line contains the
transition \verb#0 a 0#; and the third line contains the transition \verb#0 b 1#. The string has to end with a \verb1\n1.
\end{EX}
\pnsn
Next we list a few useful methods for \verb1NFA1 objects.
We assume that \texttt{a} and
\texttt{b} are automata,  and \texttt{w} is a string (a word).
\pbsn
\texttt{a.evalWordP(w)}:  returns \texttt{True} or \texttt{False}, depending on whether the automaton \texttt{a} accepts \texttt{w}.
\pssn
\texttt{a.emptyP()}: returns \texttt{True} or \texttt{False}, depending on whether the language accepted by $\texttt{a}$ is empty.
\pssn
\texttt{a \& b}: returns an NFA accepting $\lang(\texttt{a})\cap\lang(\texttt{b})$. Both
$\texttt{a}$ and $\texttt{b}$ must be NFAs with no $\ew$-transitions.
\pssn
\texttt{a.elimEpsilon()}: alters $\texttt{a}$ to an equivalent NFA with no $\ew$-transitions.
\pssn
\texttt{a.epsilonP()}: returns \texttt{True} or \texttt{False},  depending on whether
$\texttt{a}$ has any $\ew$-transitions.
\pssn
\texttt{a.makePNG(fileName='xyz')}: creates a file \verb1xyz.jpg1 containing a picture of the automaton (or transducer) $\texttt{a}$.
\pssn
\texttt{a.enumNFA(n)}: returns the set of words of length up to \texttt{n} that are accepted by the automaton $\texttt{a}$.
\begin{EX}
The following example shows a naive implementation of \verb1a.evalWordP(w)1
\begin{verbatim}
   b = fl.FL([w]).toNFA()
   c = a & b
   return not c.emptyP()
\end{verbatim}
One verifies that \texttt{w} is accepted by $\texttt{a}$ if and only if $\texttt{a}$ intersected
with the automaton accepting \{\texttt{w}\} accepts something.
\end{EX}

\section{Transducer Object Classes and Methods}\label{sec:transducers}
In this section we discuss some aspects of the implementation
of transducer objects and then continue with several important methods,
which we divide into two
subsections: product constructions and rational operations.
We discuss the method for testing
functionality in  Section~\ref{sec:nonfunc}. The module containing all that is discussed in this and the next section is called \verb1transducers.py1. We can import all transducer methods as
follows.
\begin{verbatim}
from FAdo.transducers import *
\end{verbatim}
\pssn
\emshort{Transducer objects and basic methods}
\pssn
From a mathematical point of view, a \emdef{Python dictionary} is an onto function
$\texttt{delta}:D\to R$
between two finite sets of values $D$ and $R$. One writes \verb1delta[x]1 for
the image of \texttt{delta} on \verb1x1. One can define completely the dictionary
using an expression of the form
\[
\texttt{delta = \{d1:r1, d2:r2,}\>\>\ldots\>\>\> \texttt{, dN:rN\}}
\]
which defines \verb1delta[d1$i$\verb1] = r1$i$, for all $i=1,\ldots,\mathrm{N}$. In FAdo,
the class \texttt{GFT}, for General Form Transducer, is a subclass of \texttt{NFA}.
A transducer $\trt=(Q,\al,\alD,T,I,F)$ is implemented as an object
\tti with six instance variables
\pmsi \texttt{States, Sigma, Output, delta, Initial, Final}
\pmsn
corresponding to the six components of $\trt$.
Specifically, \texttt{States} is a list of unique state names, meaning
that each state name has an index which is the position of the state name
in the list, with 0 being the first index value. The variables
\texttt{Sigma, Output, Initial} and \texttt{Final} are sets, where the latter two are
sets of state indexes.  For 
efficiency reasons, the set of transitions $T$ is implemented as a
dictionary
\pmsi   \texttt{delta:} $\{0,\ldots,n-1\} \to$ (\texttt{Sigma} $\to$ $2^{\texttt{Output} \times\{0,...,n-1\}}$),
\pssn
where $n$ is the number of states.
Thus, for any $p\in\{0,\ldots,n-1\} $, \texttt{delta[$p$]} is a  dictionary, and for any input label $x$,
 \texttt{delta[$p$][$x$]} is a set of pairs $(y,q)$ corresponding to all transitions $\{(p,x/y,q)\in T\mid y\in \texttt{Output},\> q
 \hbox{ is a state index}\}$.
\pssn
Standard form transducers are objects of the FAdo class
\texttt{SFT}, which is a subclass of \texttt{GFT}. The class \texttt{SFT} is very important from an algorithmic point
of view, as most product constructions require a transducer to be in standard form.
The conversion from GFT to SFT is done using
\pmsi
\verb1   s = t.toSFT()1
\pssn
which assigns to \tsi a new SFT object equivalent to \ttin. The implementation of Normal Form Transducers is via the FAdo class \texttt{NFT}. This form of transducers is convenient in proving mathematical statements about transducers~\cite{Sak:2009}.
\begin{EX}\label{exSuffixTran}
The following code defines a string \verb1s1 that contains a transducer
description and then constructs an SFT transducer  from that string.
The transducer, on input $x$, returns the set of all proper suffixes
of $x$---see also Fig~\ref{fig:suffix}. It has an initial state 0 and a final state 1, and deletes
at least one of the input symbols so that what gets outputted at state 1 is
a proper suffix of the input word.
\begin{verbatim}
    s = '@Transducer 1 * 0\n'\
        '0 a @epsilon 0\n'\
        '0 b @epsilon 0\n'\
        '0 a @epsilon 1\n'\
        '0 b @epsilon 1\n'\
        '1 a a 1\n'\
        '1 b b 1\n'
    t = fio.readOneFromString(s)
\end{verbatim}
\end{EX}
\begin{figure}[ht!]
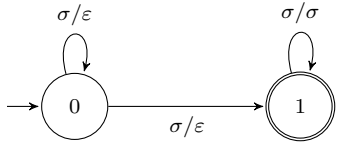

\begin{transducer}
	\node [state,initial] (q0) {$0$};
	\node [state,accepting,right of=q0] (q1) {$1$};
	\path (q0) edge [loop above] node [above] {$\sigma/\ew$} ()
		(q0) edge node [below] {$\sigma/\ew$} (q1)
		(q1) edge [loop above] node [above] {$\sigma/\sigma$} ();
\end{transducer}
\begin{center}
\parbox{0.85\textwidth}{\caption{On input $x$ the transducer shown in the figure
outputs any proper suffix of $x$.
{Note}: In this and the following transducer figures, the input and output alphabets are
equal. An arrow
with label $\sigma/\sigma$ represents a set of transitions with labels $\sigma/\sigma$, for
all alphabet symbols $\sigma$;
and similarly for an arrow with label  $\sigma/\e$. An arrow with label $\sigma/\sigma'$
represents
a set of transitions with labels $\sigma/\sigma'$ for all distinct alphabet  symbols
$\sigma,\sigma'$.}
\label{fig:suffix}
}
\end{center}
\end{figure}
Recall, for a transducer $\trt$ and word $w$, $\trt(w)$ is the set of possible outputs of $\trt$ on input $w$. Note that this set could be empty, finite, or even infinite. In any case, it is always a regular language. The FAdo
method \texttt{t.runOnWord(w)} assumes that \tti is an SFT object and returns an automaton accepting the language \ttin(w).
\begin{EX}
The following code is a continuation of Example~\ref{exSuffixTran}.
It prints the set of all proper suffixes of the word \texttt{ababb}.
\begin{verbatim}
    a = t.runOnWord('ababb')
    n = len('ababb')
    print a.enumNFA(n)
\end{verbatim}
\end{EX}
\pssn
Assuming again that \tti is an SFT object, we have the following methods.
\pssn
\texttt{t.inverse()}: returns the inverse of the transducer \ttin.
\pssn
\texttt{t.evalWordP((w,w'))}: returns \texttt{True} or  \texttt{False}, depending whether the pair \texttt{(w,w')} belongs to the relation realized by \ttin.
\pssn
\texttt{t.nonEmptyW()}: returns some word pair (\texttt{u, v}) which belongs  to the relation realized
by \ttin, if nonempty. Else, it returns the pair (\texttt{None, None}).
\pssn
\texttt{t.toInNFA()}: returns the NFA that results if we remove the output
alphabet and the output labels of the transitions in \ttin.
\pssn
\texttt{t.toOutNFA()}: returns the NFA that results if we remove the input
alphabet and the input labels of the transitions in \ttin.
%
\pbsn
\emshort{Product constructions \cite{Be:1979,Yu:handbook,Kon:2002}}
\pssn
The following methods are available in FAdo. They are adaptations of the
standard product construction \cite{HopcroftUllman} between two NFAs which produces an
NFA with transitions $((p_1,p_2),\sigma,(q_1,q_2))$,
where $(p_1,\sigma,q_1)$ and $(p_2,\sigma,q_2)$ are transitions of the two NFAs, such that the new NFA accepts the intersection of the corresponding languages. We assume that \tti and \tsi are SFT objects and \aai is an NFA
object.
\pbsn
\texttt{t.inIntersection(a)}: returns a transducer realizing all word pairs $(x,y)$ such that $x$ is accepted by \aai and $(x,y)$ is realized by \ttin.
\pssn
\texttt{t.outIntersection(a)}: returns a transducer realizing all word pairs $(x,y)$ such that $y$ is accepted by \aai and $(x,y)$ is realized by \ttin.
\pssn
\texttt{t.runOnNFA(a)}: returns the automaton accepting the language
\[
\bigcup_{w\in\lang(\aut)}\tr(w).
\]
\pssn
\texttt{t.composition(s)}: returns a transducer realizing the composition $\rel(\ttin)\circ\rel(\tsin)$
of the relations realized by the two transducers.
\pmsn
To make the presentation a little more concrete for interested readers, we comment on one of the above methods, \texttt{t.runOnNFA(a)}. The construction first considers whether \tti has any $\ew$-input transitions, that is transitions
with labels $\ew/y$. If yes, then a copy \abi of  \aai is made with $\ew$-loops added, that is,
transitions $(p,\ew,p)$ for all states $p$ in \aain. Then, if \aai has any $\ew$-transitions then a copy \tsi of \tti is made with $\ew$-loops added,
that is, transitions $(q,\ew/\ew,q)$ for all states $q$ in \ttin.
Then, the actual product construction is carried out:
if $(p,x,p')$ and $(q,x/y,q')$ are transitions of \abi and \tsin, respectively, then
$((p,q),y,(p',q'))$ is a transition of the resulting automaton.
\begin{EX}
Continuing Examples \ref{exNFA1} and \ref{exSuffixTran}, the following code
constructs an NFA object \abi accepting every word that is a proper suffix
of some word in $a^*b$. It then enters a loop that prints whether a given
string $w$ is a suffix of some word in $a^*b$.
\begin{verbatim}
    b = t.runOnNFA(a)
    while True:
        w = raw_input()
        if 'a' not in w and 'b' not in w: break
        print b.evalWordP(w)
\end{verbatim}
\end{EX}
%
\pmsn
\emshort{Rational operations \cite{Be:1979}}\pssn
A relation $\rho$ is a \emdef{rational relation}, if it is equal to $\emptyset$, or $\{(x,y)\}$
for some words $x$ and $y$, or can be obtained from other ones by using a finite number of times
any of the three (rational) operators: union, concatenation, Kleene star.
A classic result on transducers says that a relation is rational if and only if it can be realized by a transducer.
The following methods are now available in FAdo, where we assume that \tsi and \tti are SFT transducers.
\pmsn
\texttt{t.union(s):} returns a transducer realizing the union of the relations realized by \tsi and \ttin.
\pmsn
\texttt{t.concat(s):} returns a transducer realizing the concatenation of the relations realized by \tsi and~\ttin.
\pmsn
\texttt{t.star(flag=False):} returns a transducer realizing the Kleene star of the relation realized by \ttin, assuming the argument is missing or is \texttt{False}. Else it returns a transducer realizing $(\rel(\ttin))^+$.
\pmsn
The implementation of the above methods mimics the implementation
of the corresponding methods on automata.
%
%
\section{Witness of Transducer \emph{non}-functionality}\label{sec:nonfunc}
A transducer $\trt$ is called \emdef{functional} if, for every word $w$, the
set $\trt(w)$ is either empty or a singleton. A triple of words
$(w,z,z')$ is called a \emdef{witness of $\trt$'s non-functionality},
if $z\not=z'$ and $z,z'\in\trt(w)$.
In this section we present the SFT method \texttt{t.nonFunctionalW()},
which returns a witness of \ttin's non-functionality, or the triple (\texttt{None,None,None}) if \tti is functional.
The method is an adaptation of the
decision algorithms in \cite{AllMoh:2003,BeCaPrSa2003} that return whether
a given transducer in standard form is functional.
Although there are
some differences in the two algorithms, we believe that conceptually
the algorithmic technique is the same.
We first describe that algorithmic technique following the
presentation in \cite{BeCaPrSa2003}, and then we modify it in order to produce the method \texttt{t.nonFunctionalW()}. We also note that, using a careful implementation and assuming fixed alphabets, the time 
complexity of the decision
algorithm can be  quadratic with respect to the size of the transducer---see \cite{AllMoh:2003}.
\pssn
Given a standard form transducer $\trt=(Q,\al,\alD,T,I,F)$, the first phase is to construct the \emdef{square machine} $\tru$, which is defined by the following process.
\pmsi
\emshort{Phase 1}
\begin{enumerate}
  \item First define an automaton $\tru'$ as follows: states $Q\times Q$,  initial states $I\times I$, and final states  $F\times F$.
  \item If $\trt$ contains $\ew$-input transitions, that is, transitions with labels of the form $\ew/u$ then we let $\trt$ be $\trt^{\ew}$. The transitions of $\tru'$ are all the triples
    $$((p,p'),(x,x'),(q,q'))$$
    such that $(p,v/x,q)$ and $(p',v/x',q')$ are transitions of $\trt$.
  \item Return $\tru$ = a trim version of $\tru'$.
\end{enumerate}
Note that any accepting path of $\tru$ has a label
$(x_1,x_1')\cdots(x_n,x_n')$ such that the words $x_1\cdots x_n$ and $x_1'\cdots x_n'$  are outputs (possibly equal) of $\trt$ on
the \textit{same} input word. The next phase is to perform a process that starts from the
initial states and assigns a \emdef{delay value} to each state, which is either ZERO or a pair of words in $\{(\ew,\ew),(\ew,u),(u,\ew)\}$, with $u$ being nonempty. A delay $(y,y')$ on a state $(p,p')$ indicates that there is a  path in $\tru$ from $I\times I$ to $(p,p')$ whose label
is a word pair of the form $(fy,fy')$. This means that there is an input word that can take the transducer $\trt$ to state $p$ with output $fy$
and also to state $p'$ with output $fy'$. A delay ZERO at $(p,p')$
means that there is an input word that can take $\trt$ to state $p$
with output of the form $f\sigma g$ and to state $p'$ with output
of the form $f\sigma'g'$, where $\sigma$ and $\sigma'$ are distinct
alphabet symbols.
\pmsi
\emshort{Phase 2}
\begin{enumerate}
  \setcounter{enumi}{3}
  \item Assign to each initial state the delay value $(\ew,\ew)$.
  \item Starting from the initial states, visit all transitions in breadth-first search mode such that, if $(p,p')$ has the delay value $(y,y')$ and a transition $((p,p'),(x,x'),(q,q'))$ is visited, then the state $(q,q')$ gets a delay value $D$ as follows:
      \begin{itemize}
        \item If $y'x'$ is of the form $yxu$ then $D=(\ew,u)$.
        If $yx$ is of the form $y'x'u$ then $D=(u,\ew)$.
        If $y'x'=yx$ then $D=(\ew,\ew)$.
        Else, $D$ = ZERO.
      \end{itemize}
  \item The above process stops when a delay value is ZERO, or a state gets two different delay values, or every state gets one delay value.
  \item  If every state has one delay value and every final state has the delay value $(\ew,\ew)$ then return \true (the transducer is functional). Else, return \falsen.
\end{enumerate}
\pnsn
Next we present our witness version of the transducer functionality
algorithm. First, the square machine $\tru$ is revised such that
its transitions are of the form
      $$((p,p'),(v,x,x'),(q,q')),$$
that is, we now record in $\tru$ information about the common input $v$ (see Step~2 in Phase~1). Then, to each state $(q,q')$ we assign not only a
delay value but also a path value $(\alpha,\beta,\beta')$ which means
that, on input $\alpha$, the
transducer $\trt$ can reach state $q$ with output $\beta$ and also state
$q'$ with output $\beta'$---see Fig.~\ref{fig:Functionality}.
\begin{figure}[ht!]
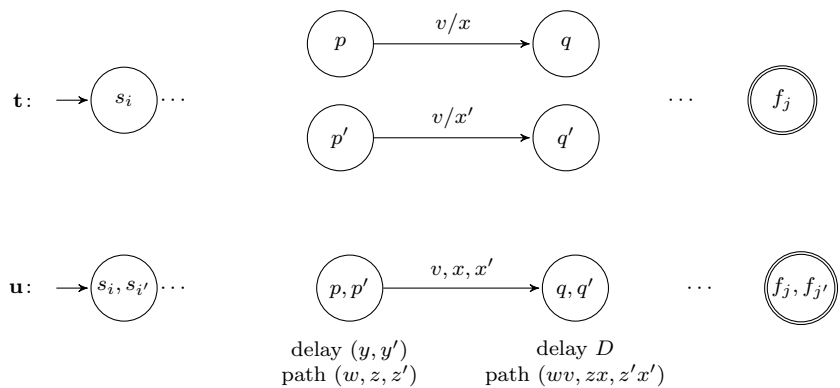

\begin{transducer}
	\node [node distance=0.75cm,above of=q0,anchor=east](refl) {};
	\node [state,initial] (q0) {$s_i$};
	\node [node distance=1cm,left=of q0,anchor=east] {$\trt\colon$};
	\node [state,right of=refl] (q1) {$p$};
	\node [state,right of=q1] (q2) {$q$};
	\node [node distance=1cm,right=of q0,anchor=east] {$\cdots$};
	\node [node distance=1.25cm,state,below of=q1] (q4) {$p'$};
	\node [state,right of=q4] (q5) {$q'$};
	\node [node distance=0.75cm,below of=q2,anchor=east](refr) {};
	\node [state,accepting,right of=refr] (q3) {$f_j$};
	\node [node distance=1cm,left=of q3,anchor=east] {$\cdots$};
	\path 
		(q1) edge node [above] {$v/x$} (q2)
		(q4) edge node [above] {$v/x'$} (q5);
	\node [node distance=2.5cm,state,initial,below of=q0] (p0) {$s_i,s_{i'}$};
	\node [node distance=1cm,left=of p0,anchor=east] {$\tru\colon$};
	\node [state,right of=p0] (p1) {$p,p'$};
	\node [state,right of=p1] (p2) {$q,q'$};
	\node [node distance=1cm,right=of p0,anchor=east] {$\cdots$};
	\node [state,accepting,right of=p2] (p3) {$f_j,f_{j'}$};
	\node [node distance=1cm,left=of p3,anchor=east] {$\cdots$};
	\path  (p1) edge node [above] {$v,x,x'$} (p2);
	\node [node distance=1cm,below of=p1](dell) {delay $(y,y')$\\path $(w,z,z')$};
	\node [node distance=1cm,below of=p2](delr) {delay $D$\\path $(wv,zx,z'x')$};
\end{transducer}
\begin{center}
\parbox{0.85\textwidth}{\caption{The figure shows the structure of the (revised) square machine $\tru$ corresponding to the
given transducer $\trt$.  The delay and path values for the states of $\tru$
are explained in Definition~\ref{def1}}
\label{fig:Functionality}}
\end{center}
\end{figure}
\begin{definition}\label{def1}
Let $(q,q')$ be a state of the new square machine $\tru$.  The set of \emdef{delay-path values} of $(q,q')$ is defined as follows.
\begin{itemize}
  \item If $(q,q')$ is an initial state then
  $((\ew,\ew),(\ew,\ew,\ew))$ is a delay-path value of $(q,q')$.
  \item If $((p,p'),(v,x,x'),(q,q'))$ is a transition in $\tru$ and $(p,p')$ has a delay-path value $(C,(w,z,z'))$, then $(D,(wv,zx,z'x'))$ is a delay-path value of $(q,q')$, where $D$ is defined as follows.
      \begin{enumerate}
        \item If $C=(y,y')\not=$  ZERO and $y'x'$ is of the form $yxu$ then $D=(\ew,u)$.
        \item If $C=(y,y')\not=$  ZERO and $yx$ is of the form $y'x'u$ then $D=(u,\ew)$.
        \item If $C=(y,y')\not=$  ZERO and $y'x'=yx$ then $D=(\ew,\ew)$
        \item Else, $D$ = ZERO.
      \end{enumerate}
\end{itemize}
For $(q,q')$, we also define a \emdef{suffix triple} $(w_{qq'},z_{qq'},z_{qq'}')$
to be the label of any path from $(q,q')$ to a final state of $\tru$.
\end{definition}
\begin{remark}\label{rem1}
The above definition implies that if a state $(p,p')$ has a
delay-path value $(C,(w,z,z'))$, then there is a path in $\tru$
whose label is $(w,z,z')$. Moreover, by the definition of $\tru$, the transducer $\trt$ on input $w$ can reach state $p$ with output $z$ and also state $p'$ with output $z'$. Thus, if $(p,p')$ is a final state, then
$z,z'\in\trt(w)$.
\end{remark}
\pssi
\emshort{Algorithm nonFunctionalW}
\begin{enumerate}
\item Define function \texttt{completePath}($q,q'$) that follows a shortest path
   from $(q,q')$ to a final state of $\tru$ and returns a suffix triple (see Definition~\ref{def1}).
  \item Construct the revised square machine $\tru$, as in Phase 1 above but now use transitions of the form
      $$((p,p'),(v,x,x'),(q,q')),$$
      (see step~2 in Phase~1).
  \item Assign to each initial state
      the delay-path value $((\ew,\ew),(\ew,\ew,\ew))$.
  \item Starting from the initial states, visit all transitions in breadth-first search mode. If $(p,p')$ has  delay-path value $((y,y'),(w,z,z'))$, and a transition $((p,p'),(v,x,x'),(q,q'))$ is visited,
      then compute the delay value $D$ of $(q,q')$ as in steps 1--4 of Definition~\ref{def1}, and
      let $R=(wv,zx,z'x')$. Then,
      \begin{enumerate}
        \item if $D$ is ZERO,  then invoke
          \texttt{completePath}$\,(q,q')$ to get a suffix triple $(w_{qq'},z_{qq'},z_{qq'}')$
         and \texttt{return} $(wvw_{qq'},zxz_{qq'},z'x'z_{qq'}')$.
        \item if $(q,q')$ is final and $D\not=(\ew,\ew)$,
          \texttt{return} $(wv,zx,z'x')$.
        \item if $(q,q')$ already has a delay value $\not=D$ and, hence, a path value $P=(w_1,z_1,z_1')$, then invoke
          \texttt{completePath}$\,(q,q')$ to get a suffix triple $(w_{qq'},z_{qq'},z_{qq'}')$.
          Then,
          \begin{itemize}
          \item
           If $zxz_{qq'}\not=z'x'z_{qq'}'$ \texttt{return} $(wvw_{qq'}, zxz_{qq'},z'x'z_{qq'}')$.
           \item
           Else \texttt{return} $(w_1w_{qq'},z_1z_{qq'},z_1'z_{qq'}')$.
           \end{itemize}
        \item else assign $(D,R)$ to $(q,q')$ as delay-path value and continue the breadth-first process.
      \end{enumerate}
\item At this point \texttt{return (None,None,None)}, as the breadth-first process has been completed.
\end{enumerate}
\emph{Terminology.} Let $A=(w_1,\ldots,w_k)$ be a tuple consisting of words. 
The \emdef{size} $\sz A$ of $A$ is the number $\sum_{i=1}^k\sz{w_i}+(k-1)$.
For example, $\sz{(0,01,10)}=7$. 
If $\{A_i\}$ is any set of word tuples then a \emdef{minimal} element
(of that set) is any $A_i$
whose size is the minimum of $\{\sz{A_i}\}$.
\begin{theorem}\label{resNonFunc}
If algorithm {\rm\texttt{nonFunctionalW}} is given as input a standard form transducer $\trt$, then it returns
either a size $O(|\trt|^2)$ witness of $\trt$'s non-functionality, or
the triple (\texttt{None,None,None}) if $\trt$ is functional.
\end{theorem}
Before we proceed with the proof of the above result, we note that 
there is a sequence $(\trt_i)$ of non-functional transducers 
such that
$\sz{\trt_i}\to\infty$ and any minimal witness of $\trt_i$'s non-functionality
is of size $\Theta(|\trt_i|^2)$. 
Indeed, let $(p_i)$ be the sequence of primes in increasing order and
consider the transducer
$\trt_i$ shown in Fig.~\ref{fig:tr:seq}. It has 
size $\Theta(p_i)$ and every output word $w$ of $\trt_i$ has length equal to that of the input used to get $w$. 
The relation realized by $\trt_i$ is
$$
\{(0^{mp_i},0^{mp_i}),\,(0^{n(p_i+1)},10^{n(p_i+1)-1})\>:\>m,n\in\N\}.
$$
Any minimal witness of $\trt_i$'s non-functionality is of the form 
$w_i=(0^{mp_i},0^{mp_i},10^{n(p_i+1)-1})$ such that $mp_i=n(p_i+1)$.
Using standard facts from number theory, we have that $n\ge p_i$.
Hence, $\sz{w_i}\ge2+3\times p_i(p_i+1)$, that is,
$\sz{w_i}=\Theta(\sz{\trt_i}^2)$.%
\begin{figure}[ht!]
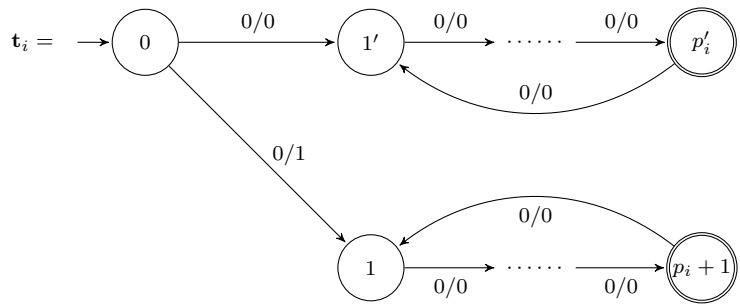

\begin{transducer}
	\node [state,initial] (q0) {$0$};
	\node [node distance=1.5cm,left of= q0] (name) {$\trt_i=$};
	\node [state, right of=q0] (q1) {$1'$};
	\node [node distance=2.20cm,right of=q1] (q2){$\cdots\cdots$};
	\node [state,accepting,node distance=2.20cm,right of=q2] (q3) {$p_i'$};
	
	\node [state,below of= q1] (q1b) {$1$};
	\node [below of=q2] (q2b) {$\cdots\cdots$};
	\node [state,accepting,below of=q3] (q3b) {$p_i+1$};
	\path 
	   (q0) edge node [above] {$0/0$} (q1)
	   (q0) edge node [right] {$\>0/1$} (q1b)
	   (q1)  edge node [above] {$0/0$} (q2)
	   (q1b) edge node [below] {$0/0$} (q2b)
	   (q2)  edge node [above] {$0/0$} (q3)
	   (q2b) edge node [below] {$0/0$} (q3b)   
	   (q3)  edge [bend left=38] node [above] {$0/0$} (q1)  
	   (q3b) edge [bend right=38] node [below] {$0/0$} (q1b)  
     ;
\end{transducer}
\begin{center}
\caption{
Transducers with quadratic size minimal witnesses of non-functionality.
\label{fig:tr:seq}}
\end{center}
\end{figure}
\pssi
The following lemma is useful for establishing the correctness of the algorithm
\texttt{nonFunctionalW}.
\begin{lemma}
If a state $(q,q')$ has a delay-path value $((s,s'),(\alpha,\beta,\beta'))$
then
there is a word $h$ such that $\beta=hs$ and $\beta'=hs'$.
\end{lemma}
\begin{proof}
We use induction based on Definition~\ref{def1}. If the
given delay-path value is
$((\ew,\ew),(\ew,\ew,\ew))$ the statement is true.
Now suppose that there is a transition $((p,p'),(v,x,x'),(q,q'))$ such that the statement is true for state $(p,p')$ (induction hypothesis) and
$((s,s'),(\alpha,\beta,\beta'))$ results from a
delay-path value $(C,(w,z,z'))$ of $(p,p')$.
As $(s,s')\not=$ ZERO then also $C\not=$ ZERO, so $C$ is of the form $(y,y')$
and one of the three cases 1--3 of Definition~\ref{def1} applies.
Moreover, by the induction hypothesis on $(p,p')$ we have
$z=gy$ and $z'=gy'$, for some word $g$, hence, $\beta=gyx$ and $\beta'=gy'x'$.
\pssn
Now we consider the three cases.
If $y'x'=yxu$ then $(s,s')=(\ew,u)$. Also,
for $h=gyx$ we have $\beta=hs$ and $\beta'=hs'$, as required.
If $yx=y'x'u$ then $(s,s')=(u,\ew)$ and one works analogously.
If $yx=y'x'$ then $(s,s')=(\ew,\ew)$. Also, $\beta=\beta'$ and
the statement follows using $h=\beta$.
\end{proof}
\pmsn
\begin{proof} \emph{(of Prop.~\ref{resNonFunc})}
First note that the algorithm returns a triple other than (\texttt{None,None,None})
exactly the {first} time when one of the following occurs (i) a ZERO value
for $D$ is computed, or (ii) a value of $D$ other than $(\ew,\ew)$ is computed for
a final state, or (iii) a value of $D$, other than the existing delay value, of a visited
state is computed. Thus, the algorithm assigns at most one delay value to
each state $(q,q')$. If the algorithm assigns exactly one delay value to each state
and terminates at step~5, then its execution is essentially the same as that of the
decision version of the algorithm, except for the fact that  in the decision
version no path values are computed. Hence, in this case the transducer is
functional and the algorithm correctly returns (\texttt{None,None,None})
in step~5.

In the sequel we assume that the algorithm terminates in one
of the three subcases (a)---(c) of step~4. So
let $(q,q')$ be a state at which the algorithm computes some delay value $D$
and path value $R=(\alpha,\beta,\beta')$---see~step 4.
It is sufficient to show the following statements.
\begin{description}
  \item[S1] If $D$ is ZERO then $(\alpha w_{qq'},\beta z_{qq'},\beta'z_{qq'}')$ is a witness of $\trt$'s non-functionality.
  \item[S2] If $(q,q')$ is final and $D\in\{(\ew,u),(u,\ew)\}$, with $u$ nonempty,  then $(\alpha,\beta,\beta')$ is a witness of $\trt$'s non-functionality.
  \item[S3] If $D$ is of the form $(s,s')$ and $((s_1,s'_1),(\alpha_1,\beta_1,\beta_1'))$ is the
  existing delay-path value of $(q,q')$ with $(s_1,s'_1)\not=(s,s')$, then one of the following triples is a witness of $\trt$'s non-functionality
      $$(\alpha w_{qq'},\beta z_{qq'},\beta'z_{qq'}'), \>\>\>
        (\alpha_1 w_{qq'},\beta_1 z_{qq'},\beta_1'z_{qq'}').$$
\end{description}
For statement S1, by Remark~\ref{rem1}, it suffices to show that $\beta z_{qq'}\not=\beta'z'_{qq'}$.
First note that 
$D$ is ZERO exactly
when there is a transition $((p,p'),(v,x,x')$, $(q,q'))$ such that
state $(p,p')$ has a delay-path value $((y,y'),(w,z,z'))$
and $yx,y'x'$ are of the form $f\sigma g$ and $f\sigma'g'$, respectively,
with $\sigma,\sigma'$ being distinct letters, and
$\alpha=wv$, $\beta=zx$, $\beta'=z'x'$. Now using the above lemma we get, for some $h$ 
\pssi
$\beta z_{qq'}=zxz_{qq'}=hyxz_{qq'}=hf\sigma gz_{qq'}$ and
\pnsi
$\beta' z'_{qq'}=z'x'z'_{qq'}=hy'x'z'_{qq'}=hf\sigma' g'z'_{qq'}$,
\pssn
which implies $\beta z_{qq'}\not=\beta'z'_{qq'}$, as required.
\pssn
For statement S2, by Remark~\ref{rem1}, it suffices to show that $\beta \not=\beta'$. By symmetry, we only consider the case of $D=(\ew,u)$. First note that 
$D$ is $(\ew,u)$ exactly
when there is a transition $((p,p'),(v,x,x'),(q,q'))$ such that
state $(p,p')$ has a delay-path value $((y,y'),(w,z,z'))$
and $y'x'=yxu$, and
$\alpha=wv$, $\beta=zx$, $\beta'=z'x'$.
Using the above lemma we get, for some $h$ 
\pssi
$\beta =zx=hyx$ and $\beta' =z'x'=hy'x'=hyxu$,
\pssn
which implies $\beta \not=\beta'$, as required.
\pssn
For statement S3, we assume that $\beta z_{qq'}=\beta'z'_{qq'}$ and we show that $\beta_1 z_{qq'}\not=\beta_1 z'_{qq'}$.
Assume for the sake of contradiction that also $\beta_1 z_{qq'}=\beta_1 z'_{qq'}$. Using the above lemma we get, for some $h$ 
\pssi
$\beta=hs,\,\beta'=hs',\,\beta_1=h_1s_1,\, \beta_1'=h_1s_1'$.
\pssn
Also by the assumptions we get $hsz_{qq'}=hs'z'_{qq'}$ and
$h_1s_1z_{qq'}=h_1s_1'z'_{qq'}$, implying that $sz_{qq'}=s'z'_{qq'}$ and
$s_1z_{qq'}=s_1'z'_{qq'}$. If $z_{qq'}=z'_{qq'}$ then $s=s'=\ew$
and $s_1=s_1'=\ew$, which is impossible as $(s,s')\not=(s_1,s_1')$.
If $z_{qq'}$ is of the form $z_1z'_{qq'}$ (or vice versa), then
we get that $(s,s')=(s_1,s_1')$, which is again impossible.
\pssn
Regarding the size of the witness returned, consider again statements S1--S3
above.
Then, the size of the witness is 
$\sz{(x,y,z)}+\sz{(w_{qq'},z_{qq'},z'_{qq'})}-2$, 
where $(w_{qq'},z_{qq'},z'_{qq'})$ could be $(\ew,\ew,\ew)$ 
and $(x,y,z)$ is a path value of state $(q,q')$: $(\alpha,\beta,\beta')$ 
or $(\alpha_1,\beta_1,\beta_1')$. As $(w_{qq'},z_{qq'},z'_{qq'})$
is based on a shortest path from $(q,q')$ to a final state of $\tru$,
we have $\sz{(w_{qq'},z_{qq'},z'_{qq'})}<\sz{\tru}$.
As the algorithm visits each transition of $\tru$ at most once, 
and $(x,y,z)$ is built by concatenating transition labels starting from
label $(\ew,\ew,\ew)$, we have that
the size of $(x,y,z)$ is bounded by the sum of the sizes of the transitions of
$\tru$. Hence, the size of the witness is $O(\sz{\tru})$. The claim about
the size of the witness follows by the fact that $\sz{\tru}=\Theta(\sz{\trt}^2)$.
\end{proof}

\section{Object Classes Representing Code Properties}\label{sec:codes}
In this section we discuss our implementation of objects representing code properties. The set of
all code properties is uncountable, but of course one can only implement countably many
properties. So we are interested in systematic methods that allow one to
 {formally} describe
code properties.
Three such formal methods are the implicational conditions of \cite{Jurg:1999}, where a property is described
by a first order formula of a certain type, the regular trajectories of \cite{Dom:2004}, where a
property is described by a regular expression over $\{0,1\}$, and the transducers of \cite{DudKon:2012},
where a property is described by a transducer. These formal methods appear to be able to describe most
properties of practical interest. The formal methods of regular trajectories and transducers are implemented here,
as the transducer formal method follows naturally our implementation of transducers, and every regular expression
of the  regular
trajectory formal method can be converted efficiently to a transducer object of the transducer formal method. The implementation
of implicational conditions is an interesting topic for future research.
\pmsn
Next we review quickly the formal methods of regular trajectories and transducers, and then
discuss our implementation of these formal methods.
\pmsn
\emshort{Regular trajectory properties \cite{Dom:2004}.}
In this formal method a regular expression $\ree$ over $\{0,1\}$ describes the  code property $\pty_{\ree}$
as follows.
The 0s in $\ree$ indicate positions of alphabet symbols that make up a word $v$, say,
and  the 1s in $\ree$ indicate positions of arbitrary symbols that, together with the 0s, make up
a word $u$, say. A language $L$ satisfies $\pty_{\ree}$ if there are no two different words in $u,v\in L$
such that $u$ has the structure indicated by $\ree$ and $v$ has a structure obtained by
deleting the 1s from $\ree$. For example, the \emdef{infix code} property is defined by
the regular expression $1^*0^*1^*$, which says that by deleting consecutive symbols at the beginning and/or
at the end of an $L$-word $u$, one cannot get a different $L$-word. Equivalently, $L$ is an infix code
if no $L$-word is an infix of another $L$-word. Note that $1^*0^*1^*$ describes all infix codes
over all possible alphabets.
\pmsn
\emshort{Input-altering transducer properties \cite{DudKon:2012}.}
A transducer $\trt$ is \emdef{input-altering} if, for all words $w$, $w\notin\trt(w)$.
In this formal method such a transducer $\trt$ describes  the code property
$\iatp_{\trt}$ consisting of all languages $L$ over the input alphabet of $\trt$ such that
\begin{equation}\label{eqIAT}
\trt(L)\cap L = \emptyset.
\end{equation}
With this formal method we can define the \emdef{suffix code} property: $L$ is a suffix code if no
$L$-word is a proper suffix of an $L$-word.
The transducer defined in Example~\ref{exSuffixTran} is input-altering and describes the
suffix code property over the alphabet \texttt{\{a, b\}}. Similarly, we can define the infix code
property by making another transducer that,
on input $w$, returns any proper infix of $w$.
We note that, for every regular expression $\ree$ over $\{0,1\}$ and alphabet $\al$, one can construct in linear time
an input-altering transducer $\trt$ with input alphabet $\al$  such that $\pty_{\ree}=\iatp_{\trt}$
\cite{DudKon:2012}. Thus, every regular trajectory property is an input-altering transducer property.
\pmsn
\emshort{Error-detecting properties via input-preserving transducers  \cite{Kon:2002,DudKon:2012}.}
A transducer $\trt$ is \emdef{input-preserving} if, for all words $w$ in the domain of $\rel(\trt)$, $w\in\trt(w)$.
Such a transducer $\trt$ is also called a
\emdef{channel transducer}, in the sense that an input message $w$ can be transmitted via $\trt$
and the output can always be $w$ (no transmission error), or a word other than $w$ (error).
In this formal method the transducer $\trt$ describes the \emdef{error-detecting for $\trt$} property $\edp_{\trt}$
consisting of all languages $L$ over the input alphabet of $\trt$ such that
\begin{equation}\label{eqIPT}
\trt(w)\cap (L-w)=\emptyset,\>\>\hbox{ for all words $w\in L$.}
\end{equation}
The term error-detecting \emshort{for} $\trt$ is used in the sense that $L$ is meant to consist of all valid
messages one can transmit via $\trt$, and $\trt$
cannot turn a valid message into a different valid message.
We note that, for every input-altering transducer $\trt$, one can make in linear time a channel
transducer $\trt'$ such that $\iatp_{\trt}=\edp_{\trt'}$ \cite{DudKon:2012}.
Thus, every input-altering transducer property is an error-detecting property.
\begin{EX}\label{exSub1}
Consider the property \emph{1-substitution error-detecting code} over \texttt{\{a, b\}},
where error means the
substitution of one symbol by another symbol. A classic characterization is that,
$L$ is such a code if and only if
the  Hamming distance between any two different words in $L$ is at least 2~\cite{Ham:1950}.
The following channel transducer defines this property---see also Fig~\ref{fig:ed}.
The transducer will substitute at most
one symbol of the input word with another symbol.
\begin{verbatim}
s1   = '@Transducer 0 1 * 0\n'\
        '0 a a 0\n'\
        '0 b b 0\n'\
        '0 b a 1\n'\
        '0 a b 1\n'\
        '1 a a 1\n'\
        '1 b b 1\n'
t1 = fio.readOneFromString(s1)
\end{verbatim}
\end{EX}
\pssn
We note that the transducer approach to defining error-detecting code properties is very powerful,
as it allows one to model insertion and deletion errors, in addition to substitution errors---see  Fig~\ref{fig:ed}. Codes for such errors are  actively under investigation---see~\cite{PAF:2013},
for instance.
\begin{figure}[ht!]
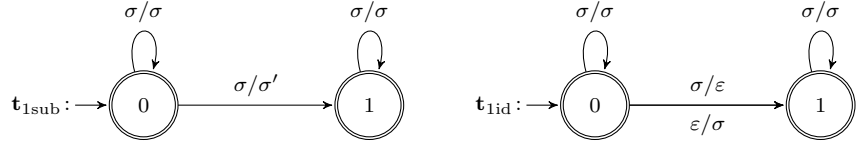

\begin{transducer}
	\node [state,initial,accepting] (q0) {$0$};
	\node [node distance=0.75cm,left=of q0,anchor=east] {$\trt_{\rm 1sub}\colon$};
	\node [state,accepting,right of=q0] (q1) {$1$};
	\node [state,initial,accepting,right of=q1] (q2) {$0$};
	\node [node distance=0.75cm,left=of q2,anchor=east] {$\trt_{\rm 1id}\colon$};
	\node [state,accepting,right of=q2] (q3) {$1$};
	\path (q0) edge [loop above] node [above] {$\sigma/\sigma$} ()
		(q0) edge node [above] {$\sigma/\sigma'$} (q1)
		(q1) edge [loop above] node [above] {$\sigma/\sigma$} ()
		(q2) edge node [above] {$\sigma/\ew$} (q3)
		(q2) edge node [below] {$\ew/\sigma$} (q3)
		(q2) edge [loop above] node [above] {$\sigma/\sigma$} ()
		(q3) edge [loop above] node [above] {$\sigma/\sigma$} ();
\end{transducer}
\begin{center}
\parbox{0.85\textwidth}{\caption{On input $x$ the transducer  $\trt_{\rm 1sub}$ 
outputs either $x$, or any word that results by substituting exactly one symbol in $x$.
On input $x$ the transducer  $\trt_{\rm 1id}$ 
outputs either $x$, or any word that results by deleting, or inserting, exactly one symbol in $x$.
{Note}: The use of labels on arrows is explained in Fig.~\ref{fig:suffix}. 
}
\label{fig:ed}}
\end{center}
\end{figure}
\pmsn
\emshort{Error-correcting properties via input-preserving transducers.}
An input-preserving (or channel) transducer is used to describe the
\emdef{error-correcting for $\trt$} property $\ecp_{\trt}$ consisting of all languages $L$
over the input alphabet of $\trt$ such that
\begin{equation}\label{eqEC}
\trt(v)\cap\trt(w)=\emptyset,\>\>\hbox{ for all distinct words $v,w\in L$.}
\end{equation}
The term error-correcting \emshort{for} $\trt$ is used in the sense that any
message $w'$ received from an $L$-word via $\trt$ can result from only one
such $L$-word, so if $w'\notin L$ then $w'$ can in principle be corrected to exactly one
$L$-word.
\begin{remark}\label{remECvsED}
It can be shown that a language is error-correcting for $\trt$
if and only if it is error-detecting for any transducer realizing the
relation $\rel(\trt^{-1})\compose\rel(\trt)$.
\end{remark}
\begin{remark}\label{rem3indep}
	All input-altering, error-detecting and error-correcting properties are 3-independences.
\end{remark}
%
%
\subsection{Implementation in FAdo.}
We present now our implementation of the previously mentioned
code properties.
We have defined the Python classes \pmsi
\texttt{TrajProp, IATProp, ErrDetectProp, ErrCorrectProp}
\pmsn
corresponding to the four types of properties discussed above.
These four property types are described, respectively,
by regular trajectory expressions, input-altering transducers,
input-preserving transducers, and input-preserving transducers.
In all four cases,
given a  transducer object,  an object of the class is created.
An object \ppi of the class \texttt{IATProp}, say, is defined via some transducer \tti and represents
a particular code property, that is, the class of languages satisfying Eq.~(\ref{eqIAT}).
\pssn
The class \texttt{ErrDetectProp}
is a 
{superclass} of the others. These classes and all related
methods and functions are in the module \texttt{codes.py} and can be imported as follows.
\begin{verbatim}
import FAdo.codes as codes
\end{verbatim}
Although each of the above four classes requires a transducer to create an object of the class,
we have defined a set of what we call \emdef{build functions} as a user interface for
creating code property objects.
These  build functions are shown next in use with specific arguments from previous examples.
\begin{EX}\label{exBuildFunctions}
Consider again Examples
\ref{exSuffixTran} and \ref{exSub1} in which the strings \texttt{s} and \texttt{s1} are defined containing, respectively, the
proper suffixes transducer and the transducer permitting up to 1 substitution error.
The following object definitions are possible with the FAdo package
\begin{verbatim}
    icp = codes.buildTrajPropS('1*0*1*', {'a', 'b'})
    scp = codes.buildIATPropS(s)
    s1dp = codes.buildErrorDetectPropS(s1)
    s1cp = codes.buildErrorCorrectPropS(s1)
    pcp = codes.buildPrefixProperty({'a', 'b'})
    icp2 = codes.buildInfixProperty({'a', 'b'})
\end{verbatim}
In the first statement, \texttt{icp} represents the infix code property over
the alphabet \texttt{\{a, b\}} and is defined via the trajectory expression \texttt{1*0*1*}.
In the next three statements, \texttt{scp, s1dp, s1cp} represent, respectively, the suffix code
property, the 1-substitution error-detecting property and the
1-substitution error-correcting property. The last two statements are explained 
below---\texttt{pcp} and \texttt{icp2} represent the prefix code  and infix code properties, 
respectively.
\end{EX}
\pmsn
\emshort{Fixed properties.}
Some properties are well known in the theory of codes, so we have created specific
classes for these properties and, therefore, FAdo users need not write
transducers, or trajectory regular expressions, for creating  these properties.
As before, users need only to know about the \texttt{build}-interfaces for
creating objects of these classes.
\pmsn
\texttt{buildPrefixProperty(Sigma)}: returns an object of the class \texttt{PrefixProp} that represents all prefix codes over the alphabet \texttt{Sigma}.
\pmsn
\texttt{buildSuffixProperty(Sigma)}: returns an object of the class \texttt{SuffixProp} that represents all suffix codes over the alphabet \texttt{Sigma}.
\pmsn
\texttt{buildInfixProperty(Sigma)}: returns an object of the class \texttt{InfixProp} that represents all infix codes over the alphabet
\texttt{Sigma}. Note that an infix code is both prefix and suffix and this fact is reflected in the class definition, by considering \texttt{InfixProp} a Python subclass of both \texttt{PrefixProp} and \texttt{SuffixProp}.
\pmsn
\texttt{buildOutfixProperty(Sigma)}: returns an object of the class \texttt{OutfixProp} that represents all outfix codes over the alphabet
\texttt{Sigma}. A language $L$ is an \emdef{outfix code} if deleting an infix of an $L$-word cannot result
into another $L$-word. Note that an outfix code is both prefix and suffix, and, as above, \texttt{OutfixProp} is a subclass of the correspondent Python classes.
\pmsn
\texttt{buildHypercodeProperty(Sigma)}: returns an object of the class \texttt{HypercodeProp} that represents all hypercodes over the alphabet
\texttt{Sigma}. A language $L$ is a \emdef{hypercode} if deleting any symbols of an $L$-word cannot result
into another $L$-word. Note that a hypercode is both infix and outfix and as above, \texttt{HypercodeProp} is a subclass of the correspondent Python classes.
\pmsn
Each of the above methods creates internally a transducer whose input and output alphabets are equal to the
given alphabet
\texttt{Sigma}, and then passes the transducer to the constructor of the  respective class.
\pmsn
\emshort{Combining code properties.}
In practice it is desirable to be able to talk about languages that satisfy more than one
code property. For example, most of the 1-substitution error-detecting codes used in practice
are infix codes
(in fact \emdef{block codes}, that is, those whose words are of the same length).
We have defined the operation \verb1&1 between any two  error-detecting properties
independently of how they were created. This operation returns an object
representing the class of all languages satisfying both properties.
This object is constructed via the transducer that results by taking the union of the two transducers
describing the two properties---see Rational Operations in Section~\ref{sec:transducers}.
\begin{EX}
Using the properties \texttt{icp, s1dp} created above in Example~\ref{exBuildFunctions}, 
we can create the conjunction \verb2p12 of these properties, and using the properties 
\texttt{pcp, scp} we can create their conjunction \texttt{bcp} which is known as the 
\emdef{bifix code property}.
\begin{verbatim}
    p1 = icp & s1dp
    bcp = pcp & scp
\end{verbatim}
\pnsn
The object \verb2p12 represents the
property $\cP_{\ree}\cap\edp_{\texttt{s1}}$, where $\ree=\texttt{1*0*1*}$.
It is of type \texttt{ErrDetectProp}.
If, however, the two properties involved are input-altering then our implementation makes sure that
the object returned is also of type input-altering---this is the case for \texttt{bcp}.
This is important as the satisfaction problem for
input-altering transducer properties can be solved more efficiently than the
satisfaction problem for error-detecting properties---see Section~\ref{sec:code-methods}.
\end{EX}
%
%
%
\subsection{Aspects of Code Hierarchy Implementation}\label{sec:hierarchy}
As stated above, our top Python superclass is \texttt{ErrDetectProp}.
When viewed as a set of (potential) objects, this class implements the set of properties
\begin{equation}\label{eqEDP}
\edp \>=\> 
\{\edp_{\trt}\mid \trt\hbox{ is an input-preserving transducer}\}.
\end{equation}
For any \texttt{ErrDetectProp} object \ttp, let us denote by \rep{\ttp} the property in $\edp$ represented by \ttp. If \ttps and \ttqs are any \texttt{ErrDetectProp} objects such that $\rep{\ttp}\sse\rep{\ttq}$ and we know that a language satisfies $\rep{\ttp}$ then it follows logically that the language also satisfies $\rep{\ttq}$ and, therefore, one does not need to execute the method \texttt{q.notSatsfiesW} on the automaton accepting that language. Similarly,  as $\rep{\ttp\&\ttq}=\rep{\ttp}$ the method invocation 
$\ttp\&\ttq$ should simply return \ttp. 
It is therefore desirable to have a method `$\po$' such that if $\ttp\po\ttq$ returns true then $\rep{\ttp}\sse\rep{\ttq}$.

In fact, for \texttt{ErrDetectProp} objects, we have implemented the methods
`$\&$' and `$\po$' in a way that the triple
$(\texttt{ErrDetectProp},\&,\po)$ constitutes a syntactic hierarchy (see further below) which can
be used to simulate the properties in~(\ref{eqEDP}).
In practice this means that `$\&$' simulates intersection between properties and
`$\po$' simulates subset relationship between two properties such that the
following desirable statements hold true, for any  \texttt{ErrDetectProp} objects \texttt{p, q}
\pmsi  \verb1p & p1 returns \verb1p1
\pssi   $\verb1p1 \po \verb1q 1$ if and only if \verb1 p & q 1 returns \verb1p1
\pmsn
We note that the syntactic simulation of the properties in Eq.~\eqref{eqEDP} is not complete
(in fact it 
{cannot} be complete): for any \texttt{ErrDetectProp} objects \texttt{p, q}
defined via transducers $\trt$ and $\trs$ with $\edp_{\trt}\sse\edp_{\trs}$ it does not
always hold that $\verb1p1 \po \verb1q1$. On the other hand, our implementation of
the set of the five fixed properties constitutes a complete simulation of these properties,
when the same alphabet is used. This implies, for instance, that
\pmsi  \verb1pcp & icp2 1  returns  \verb1 icp21
\pmsn
where we have used the notation of Example~\ref{exBuildFunctions}.
Our implementation associates to each  
object \ttps of type \texttt{ErrDetectProp} a nonempty set 
\ttp.ID of names. If \ttps is a fixed property object, \ttp.ID has one hardcoded name.
If \ttps is built from a transducer $\trt$, \ttp.ID has one name, the name of $\trt$---this name is based on a string description of $\trt$. If \ttps = \verb1q&r1, then \ttp.ID is the union of
\ttq.ID and \texttt{r}.ID minus any fixed property name $N$ for which another fixed property name $M$ exists in the union 
such that the $M$-property is contained in the $N$-property.
%
%
\pssi
Next we give a mathematical definition of what it means to simulate a set of
code properties $\psetq=\{\cQ_j\mid j\in J\}$ via a syntactic hierarchy
$(G,\sop,\po)$---see definition below---which can  ultimately be implemented
(as is the case here) in a standard programming language.
The idea is that each $g\in G$ represents a property $\rep{g}=\cQ_j$, for some index $j$,
and $G$ is the set of generators of the semigroup $(\gene{G},\sop)$ whose operation `$\sop$'
simulates the process of combining properties  in $\psetq$, that is $\rep{x\sop y}=\rep x\cap\rep y$,
and the partial order  `$\po$' simulates subset
relation between properties, that is $x\po y$ implies $\rep x\sse\rep y$, for all $x,y\in\gene G$.
\pnsi
The first  result is that there is an efficient simulation of the set
$\edp$ in Eq.~\eqref{eqEDP}---see Theorem~\ref{th:sim}.
The second result is that there can be no \emph{complete} simulation of that set
of properties, that is, a simulation such that $\rep x\sse\rep y$ implies $x\po y$,
for all $x,y\in\gene G$---see Theorem~\ref{th:nonsim}.
\pssn
\begin{definition}\label{def:sh}
A syntactic hierarchy is a triple $(G,\sop,\po)$ such that $G$ is a  nonempty set  and
\begin{enumerate}
\item
$(\gene{G},\sop)$ is the commutative semigroup generated by $G$ with computable
operation `$\sop$'.
\item
$(\gene{G},\po)$ is a decidable partial order (reflexive, transitive, antisymmetric).
\item
For all $x,y\in\gene G$, $x\po y$ implies $x\sop y=x$.
\item
For all $x,y\in\gene G$, $x\sop y\po x$.
\end{enumerate}
\end{definition}
\pssn

Next we list a few properties of the operation `$\sop$' and the 
order `$\po$'. \pssn
\begin{lemma}\label{lemSH}
The following statements hold true, for all $x,y,z\in\gene G$,
\begin{enumerate}
\item
$x\po x$ and $x\sop x = x$
\item
$x\po y$ if and only if $x=y\sop z$ for some $z\in\gene G$.
\item
$x=x\sop y$ if and only if $x\po y$.
\item
If $x\le y$ and $x\le z$ then $x\le y\sop z$.
\item
If $x=g_1\sop\cdots\sop g_n$, for some $g_1,\ldots,g_n\in G$, with all $g_i$'s
distinct and $n\ge2$, then $x<g_1$ or $x<g_2$, and hence $x$ is not maximal.
\item
$x$ is maximal if and only if $x$ is prime (meaning, $x=u\sop v$ implies $x=u=v$).
\end{enumerate}
\end{lemma}
\begin{proof}
The proof of correctness is based on the previous definition and uses standard logical arguments. We present only proofs for the second and fourth statements.
\pnsi
The `if' part of the second statement follows from the fourth statement of the above definition, and the `only if' part follows from the third statement of the above definition.
\pnsi
For the fourth statement, using the fact that $x\sop(y\sop z)\po y\sop z$,
it is sufficient to show that $x=x\sop(y\sop z)$. This follows when we
note that $x\po y$ implies $x=x\sop y$ and $x\po z$ implies $x=x\sop z$.
\end{proof}
\pssn
\begin{definition}\label{def:sim}
Let $\psetq=\{\cQ_j\mid j\in J\}$ be a set of properties, for some index set
$J$. A (syntactic) simulation of $\psetq$ is a quintuple $(G,\sop,\po,\rep{\>},\phi)$
such that $(G,\sop,\po)$ is a syntactic hierarchy and
\begin{enumerate}
\item
$\rep{\>}$ is a surjective mapping of $\gene G$ onto $\psetq$;
\item
for all $x,y\in\gene G$, $x\po y$ implies $\rep x\sse \rep y$;
\item
for all $x,y\in\gene G$, $\rep{x\sop y}=\rep x\cap\rep y$;
\item
$\phi$ is a computable function of $J$ into $\gene G$ such that $\rep{\phi(j)}=\cQ_j$.
\end{enumerate}
The simulation is called \emph{complete} if, for all $x,y$
\[
\rep x\sse \rep y\quad\hbox{implies}\quad x\po y.
\]
The simulation is called linear if $J$ has a size function $\sz{\cdot}$ and $\gene G$ has a size
function $\szg{\cdot}$
such that  $\szg{\phi(j)}=O(\sz{j})$, for all $j\in J$, and
for all $x,y$
\[
\szg{x\sop y}\>=\>O(\szg{x}+\szg{y}).
\]
\end{definition}
By a size function on a set $X$, we mean any function $f$ of $X$ into
$\N_0$.
\begin{theorem}\label{th:sim}
There is a linear simulation of the set
$$\{\edp_{\trt}\mid \trt\hbox{ is an input-preserving transducer}\}.$$
\end{theorem}
\begin{proof}
Let $\mathbf T$ be the set consisting of all finite sets of transducers.
Let $T_1,T_2\in\mathbf T$.
We define
\pmsi
$G = \{\{\trt\}\mid \trt\hbox{ is an input-preserving transducer}\}$.
\pmsi
$T_1\sop T_2=T_1\cup T_2$.
\pmsi
$T_1\po T_2$, if $T_2\sse T_1$.
\pssn
The above definitions imply that $\gene G$ consists of all $T$, where $T$
is a finite nonempty set of input-preserving transducers, and that indeed $(\gene G,\sop)$ is
a commutative semigroup and $(\gene G,\po)$ is a partial order. Moreover one verifies
that the last two requirements of Definition~\ref{def:sh} are satisfied. Thus
$(G,\sop,\po)$ is a syntactic hierarchy.
\pmsn
Next we define the size function. For a finite nonempty set $T=\{\trt_1,\ldots,\trt_n\}$
of transducers, we
denote with $\lor T$ the transducer $\trt_1\lor\cdots\lor\trt_n$
of size $O(\sum_1^n\sz{\trt_i})$ realizing $\rel(\trt_1)\cup\cdots\cup\rel(\trt_n)$. Then, 
define $\szg{T} = \sz{\lor T}$.
One verifies that  $\szg{T_1\sop T_2}=O(\szg{T_1}+\szg{T_2})$ as required.
\pssn
Next we use the syntactic hierarchy  $(G,\sop,\po)$ to define the required simulation. 
First, define $\phi(\trt)=\{\trt\}$, for any input-preserving transducer $\trt$.
Then, define
\[
\rep{T} = \edp_{\lor T},\>\hbox{ which equals } \>\bigcap_{\trt\in T}\edp_{\trt}.
\]
One verifies that the requirements of Definition~\ref{def:sim} are satisfied.
\end{proof}

The next result shows that there can be no complete simulation of
the set of error-detecting properties. This follows from the undecidability
of the Post Correspondence problem using 
methods from establishing the undecidability of basic transducer related problems \cite{Be:1979}.

\begin{theorem}\label{th:nonsim}
There is no complete simulation of the set of properties
$$\{\edp_{\trt}\mid \trt\hbox{ is an input-preserving transducer}\}.$$
\end{theorem}

Before we present the proof, we establish a few  necessary auxiliary results.

\begin{lemma}\label{lem:pties}
For any input-preserving transducers $\trt,\trs$ we have
\[
\edp_{\trt}\sse\edp_{\trs}\>\hbox{ if and only if }\>
\rel(\trs\lor\trs^{-1})\sse\rel(\trt\lor\trt^{-1}).
\]
\end{lemma}
\begin{proof}
First note that Eq.~(\ref{eqIPT}) is equivalent to
\[
(w,v)\in\rel(\trt)\>\hbox{ implies }\>w=v,\>\>\hbox{ for all $v,w\in L$.}
\]
This implies that, for all input-preserving transducers $\trt,\trs$, we have
$\edp_{\trs}=\edp_{\trs^{-1}}$ and
\begin{eqnarray}
\rel(\trs)\sse\rel(\trt)&\hbox{ implies }&\edp_{\trt}\sse\edp_{\trs},\\
\edp_{\trs}&=& \edp_{\trs\lor\trs^{-1}}.
\end{eqnarray}
The statement of the lemma follows from the above observations using
standard set theoretic arguments.
\end{proof}

\begin{lemma}\label{lem:undecide}
The following problem is undecidable.
\pssi Input: input-preserving transducers $\trt,\trs$
\pssi Return: whether $\rel(\trs\lor\trs^{-1})\sse\rel(\trt\lor\trt^{-1})$
\end{lemma}
\begin{proof}
We reduce the Post Correspondence Problem (PCP) to the given problem.
Consider any instance $((u_i)_1^p,\,(v_i)_1^p)$ of PCP which is a pair
of sequences of $p$ nonempty words over some alphabet $B$, for some positive integer $p$.
As mentioned before, we use tools that have been used in showing the
undecidability of basic transducer problems \cite{Be:1979}. In particular,
we have that the given instance is a YES instance if and only if $U^+\cap V^+\not=\es$,
where
\[
U=\{(ab^i,u_i)\mid i=1,\ldots,p\}\>\hbox{ and }\>
V=\{(ab^i,v_i)\mid i=1,\ldots,p\},
\]
and we make no assumption about the intersection of the alphabets $B$ and $\{a,b\}$.
Here we define the following objects
\begin{eqnarray}
C &=& \{ab,ab^2,\ldots,ab^p\}\\
\diag(L) &=& \{(x,x)\mid x\in L\},\>\hbox{for any language $L$}\\
D &=& \diag(C^+)\cup\diag(aaB^+)\\
X &=& (\ew,aa)U^+\cup D\\
Y &=& ((C^+\times aaB^+)-(\ew,aa)V^+)\cup D\\
\trt &=& \>\hbox{ any transducer realizing X}\\
\trs &=& \>\hbox{ any transducer realizing Y}
\end{eqnarray}
We use rational relation theory to show the following claims, which establish the
required reduction from PCP to the given problem.
\pssi C1: $X$ and $Y$ are rational relations
\pssi C2: $\trt$ and $\trs$ are input preserving
\pssi C3: $U^+\cap V^+\not=\es\>$ if and only if $(X\cup X^{-1})\sse(Y\cup Y^{-1})$
\pssn
\emph{Claim C1:}
First note that, as $C^+$ and $aaB^+$ are regular languages, we have
that $D$ is a rational relation. In \cite{Be:1979}, the author shows that $U^+$ and
$((C^+\times B^+)-V^+)$ are rational relations. It follows then that $X$ is a rational relation.
Similarly, rational is also the relation
\[
(\ew,aa)((C^+\times B^+)-V^+)=((C^+\times aaB^+)-(\ew,aa)V^+),
\]
which implies that  $Y$ is rational as well.
\pssn
\emph{Claim C2:}
From the previous claim there are transducers (in fact effectively constructible)
$\trt$ and $\trs$ realizing $X$ and $Y$, respectively. Note that the domains of
both transducers are equal to $C^+\cup aaB^+$. The fact that $(x,x)\in\rel(\trt)\cap\rel(\trs)$
for all $x\in C^+\cup aaB^+$
implies that both transducers are indeed input-preserving.
\pssn
\emph{Claim C3:}
First it is easy to confirm that $X^{-1}\sse Y^{-1}$ if and only if $X\sse Y$ (in fact for
any relations $X$ and $Y$). Also the fact that $C^+\cap aaB^+=\es$ implies that
$(\ew,aa)U^+$ is disjoint from the sets
\[
((\ew,aa)U^+)^{-1},\>D,\>((C^+\times aaB^+)-(\ew,aa)V^+)^{-1}
\]
and similarly $((C^+\times aaB^+)-(\ew,aa)V^+)$ is disjoint from the same sets.
The above observations imply that
\pmsi $(X\cup X^{-1})\sse(Y\cup Y^{-1})$ if and only if
\pssi $(\ew,aa)U^+\sse ((C^+\times aaB^+)-(\ew,aa)V^+)$ if and only if
\pssi $(\ew,aa)U^+ \cap (\{a,b\}^*\times B^*)-((C^+\times aaB^+)-(\ew,aa)V^+)=\es$ if and only if
\pssi $(\ew,aa)U^+ \cap (\ew,aa)V^+=\es$ if and only if
\pssi $U^+\cap V^+=\es$, as required.
\end{proof}

\begin{proof} (of Theorem~\ref{th:nonsim}.)
For the sake of contradiction, assume there is a complete simulation $(G,\sop,\po,\rep{\>},\phi)$ of the given set of properties. Then, for
all $x,y\in\gene G$, we have
\[
\rep x\sse\rep y\>\hbox{ implies }\> x\po y.
\]
We derive a contradiction by showing that the problem in Lemma~\ref{lem:undecide}
is decidable as follows.
\pmsi 1. Let $x = \phi(\trt)$
\pssi 2. Let $y = \phi(\trs)$
\pssi 3. if $y\po x$: return YES
\pssi 4. else: return NO
\pmsn
The correctness of the `if' clause is established as follows: $y\po x$ implies
$\rep y\sse\rep x$, which implies $\edp_{\trt}\sse\edp_{\trs}$, and then
$\rel(\trs\lor\trs^{-1})\sse\rel(\trt\lor\trt^{-1})$, as required.
The correctness of the `else' clause is established as follows:
first note $y\not\po x$. We need to show $\rel(\trs\lor\trs^{-1})\not\sse\rel(\trt\lor\trt^{-1})$.
For the sake of contradiction, assume the opposite. Then, $\edp_{\trt}\sse\edp_{\trs}$,
which implies $\rep y\sse\rep x$, and then (by completeness) $y\po x$, which is a contradiction.
\end{proof}

\section{Methods of Code Property Objects}\label{sec:code-methods}
%
In the context of the research on code properties, we consider the following three algorithmic
problems as fundamental.
\begin{description}
\item[\emdef{Satisfaction problem}.] Given the description of a code property and the description of a language, decide whether
the language satisfies the property. In the \emdef{witness version} of this problem, a negative answer is also
accompanied by  an appropriate set of  words showing how the property is violated.
\item[\emdef{Maximality problem}.] Given the description of a code property and the description of a language $L$,
decide whether the language is maximal with respect to the property. In fact we allow the more general problem,
where the input includes also the description of a second language $M$ and the question is whether there is no
word $w\in M\setminus L$ such that $L\cup w$ satisfies the property. The default case is when $M=\al^*$.
In the \emdef{witness version} of this problem, a negative answer is also
accompanied by any word $w$ that can be added to the language $L$.
\item[\emdef{Construction problem}.] Given the description of a code property and two positive integers $n$ and $\ell$, construct
a language that satisfies the property and contains $n$ words of length $\ell$ (if possible).
\end{description}
It is assumed that the code property can be implemented as \ppi  via a transducer $\trt$
(whether input-altering or input-preserving)
and, in the first two problems, the language is given via an NFA $\aut$.
In the maximality problem, the second language $M$ is given via
an NFA $\autb$. In fact one can give a language via a
regular expression, in which case it is converted to an automaton.
Next we present the implementation of methods for the satisfaction and maximality problems.
We discuss briefly the construction problem in the last section of the paper.
\pbsn
\emshort{Methods} \texttt{p.satisfiesP(a)}
\pmsn
The satisfaction problem is decidable in time $O(\sz{\trt}\sz{\aut}^2)$, if  the property \ppi is of the
input-altering transducer type. This follows from Eq.~(\ref{eqIAT}) and can be implemented 
as follows
\begin{verbatim}
    c = t.runOnNFA(a)
    return (a & c).emptyP()
\end{verbatim}
For the case of \ppi being the error-detecting property, the  transducer $\trt$ is input-preserving and
Eq.~(\ref{eqIPT}) is decided via a transducer functionality test~\cite{Kon:2002}.
In FAdo this test can be implemented as follows, where the method \texttt{functionalP()}
returns whether the transducer is functional.
\begin{verbatim}
    s = t.inIntersection(a)
    return s.outIntersection(a).functionalP()
\end{verbatim}
For the case of \ppi being the error-correcting property, again the given transducer is input-preserving and
Eq.~(\ref{eqEC}) is decided via a transducer functionality test~\cite{Kon:2002}.
In FAdo this test can be implemented as follows
\begin{verbatim}
    s = t.inverse()
    return s.outIntersection(a).functionalP()
\end{verbatim}
\pmsn
\emshort{Method} \texttt{p.maximalP(a, b)}
\pmsn
The maximality problem is decidable but PSPACE-hard \cite{DudKon:2012}. In particular, for the case of
both input-altering transducer and error-detecting properties, the decision algorithm is very simple:
the language $\lang(\aut)$
is \ppin-maximal  (within $\lang(\autb)$) if and only if
\begin{equation}\label{eq-max}
\lang(\autb)\setminus (\lang(\aut)\cup\trt(\aut)\cup\trt^{-1}(\aut))=\emptyset.
\end{equation}
The above emptiness test is implemented in the method \texttt{p.maximalP(a, b)}, which returns \true or \falsen,
and uses standard NFA methods as well as the transducer methods \texttt{t.inverse()} and
\texttt{t.runOnNFA(a)}.
For the case of error-correcting properties, our implementation makes  use of Remark~\ref{remECvsED}.
\pbsn
\emshort{Methods with witnesses:} \texttt{p.notMaximalW(a, b)}
\pmsn
It can be shown that any word belonging to the set shown in Eq.~(\ref{eq-max}) can be added to $\lang(\aut)$ and
the resulting language will still satisfy the property \cite{DudKon:2012}. This word can serve as a witness
of the non-maximality of $\lang(\aut)$. If no such word exists, the method returns \nonen.
\pbsn
\emshort{Methods with witnesses:} \texttt{p.notSatisfiesW(a)}
\pmsn
For input-altering transducer and error-detecting properties,
the witness version of the method \texttt{p.satisfiesP(a)}
 returns either a pair of {\em different} words
$u,v\in\lang(\aain)$
violating the property, that is, $v\in\ttin(u)$ or $u\in\ttin(v)$, or they return the pair $(\nonen, \nonen)$.
In the former case, the pair $(u,v)$ is called a
\emdef{witness of the non-satisfaction of \ppi by} the language $\lang(\aain)$.
For  error-correcting properties \ppin, a witness of non-satisfaction by $\lang(\aain)$ is a triple
of words $(z,u,v)$ such that $u,v\in\lang(\aain)$ and $u\not=v$ and $z\in\ttin(u)\cap\ttin(v)$.
Next we discuss how to accomplish this by
changing the implementations of \texttt{p.satisfiesP(a)} shown before.
\pmsn
\begin{description}
\item[Case 1:] For  input-altering transducer properties, we have the Python code
\begin{verbatim}
    return t.inIntersection(a).outIntersection(a).nonEmptyW()
\end{verbatim}
Recall from  Section~\ref{sec:transducers}, the above returns  (if possible) a witness for the nonemptiness of
the transducer \tti when the input and output parts of \tti are intersected by the language $\lang(\aain)$.
This witness corresponds to the emptiness test in Eq.~(\ref{eqIAT}), as required.
\pmsn
\item[Case 2:] For error-detecting properties, the defining transducer is a channel (input-preserving)
and, therefore, we use the method \texttt{nonFunctionalW()} instead of \texttt{nonEmptyW()}.
More specifically, we use the code
\begin{verbatim}
  u, v, w = t.inIntersection(a).outIntersection(a).nonFunctionalW()
  if u == v:
      return u, w
  else:
      return u, v
\end{verbatim}
If  the method \texttt{nonFunctionalW()} returns a triple of words
$(u,v,w)$ then, by Proposition~\ref{resNonFunc} and the definitions of the
\texttt{in/out} intersection methods,
we have that $v\not=w$, $v\in\ttin(u)$, $w\in\ttin(u)$ and all three words are in $\lang(\aain)$.
This implies that at least one of $u\not=v$ and $u\not=w$ must be true and, therefore,
the returned value is the pair $(u,v)$ or $(u,w)$. Moreover, the returned pair indeed violates the
property. Conversely, if the non-functionality method returns a triple of \nonen{}s then the constructed
transducer is not functional.
Then $\lang(\aain)$ must satisfy the property, otherwise any different words $v,w\in\lang(\aain)$ violating the
property could be used to make the triple $(v,v,w)$, or $(w,w,v)$,  which would serve as a witness of the non-functionality
of the transducer.
\pmsn
\item[Case 3:] For error-correcting properties, we use again the
non-functionality witness method. 
\begin{verbatim}
    return t.inverse().outIntersection(a).nonFunctionalW()
\end{verbatim}
For the correctness of this algorithm, one argues similarly as in the previous case.
\end{description}
 \pmsn
The above discussion establishes the following consequence of Proposition~\ref{resNonFunc} and
the definitions of product constructions in Section~\ref{sec:transducers}.
\begin{proposition}
The algorithms implemented in the three forms (input-altering transducer, error-detecting, error-correcting) of
the method {\rm \texttt{p.notSatisfiesW(a)}} correctly
returns a witness of the non-satisfaction of the  property {\rm\ppin} by the language {\rm $\lang(\aain)$}.
\end{proposition}
\begin{EX}\label{exIPython1}
The following Python interaction shows that the language $a^*b$ is a prefix  and
1-error-detecting code. Recall from previous examples that the Python strings \texttt{st, s1}
contain, respectively, the descriptions of an NFA accepting $a^*b$, and
a channel transducer that performs up to one substitution error when  given an input word.
\begin{verbatim}
>>> a = fio.readOneFromString(st)
>>> pcp = codes.buildPrefixProperty({'a','b'})
>>> s1dp = codes.buildErrDetectPropS(s1)
>>> p2 = pcp & s1dp
>>> p2.notSatisfiesW(a)
(None, None)
\end{verbatim}
\end{EX}

\section{Uniquely Decodable/Decipherable Codes}\label{sec:UDCode}
The property of unique decodability or decipherability, \emdef{UD code property} for short, is probably the
first historically property of interest in coding theory from the points of view of both
information theory \cite{Sardinas:Patterson}
as well as formal language theory \cite{Sch:1955,Niv:1966}.
A language $L$ is said to be a UD code if, for any two sequences
$(u_i)_1^n$ and $(v_j)_1^m$ of $L$-words such that $u_1\cdots u_n=v_1\cdots v_m$, we have that
$n=m$ and the two sequences are identical. In simpler terms, every word in $L^*$ can be
decomposed  in exactly one way
as a sequence of $L$-words. In this section we describe our implementation of the satisfaction and maximality methods for the UD~code property, as the corresponding methods for error-detecting properties discussed above are not applicable to the UD~code property. 
\begin{remark}\label{rem:UDcodes}
	In \cite{JuKo:handbook}, as a consequence of a more general result, it is shown that the UD~code property is not an $n$-independence for any $n<\aleph_0$. Thus, this property is not an error-detecting property, so it is not describable by any input-preserving transducer. A specific argument showing that this property is not a 3-independence is as follows: the language $L=\{a,ab,ba\}$ is not a UD~code, as $a(ba)=(ab)a$, but every subset     of $L$ having $<3$ elements is a UD~code.
\end{remark}
We assume that \aai is an NFA object without $\ew$-transitions.
One can create the UD code property using the following syntax
\begin{verbatim}
    p = codes.buildUDCodeProperty(a.Sigma)
\end{verbatim}
\pssn
\emshort{The method} \texttt{p.notSatisfiesW(a)}\pnsn
The satisfaction problem for this property was discussed first in the well known paper
\cite{Sardinas:Patterson}, where the authors produce a necessary and sufficient condition
for a finite language to be a UD code---we note that some feel that that condition does not
clearly give an algorithm, as for instance the papers \cite{Markov:62,Levenshtein:62}.
Over the years people have established faster algorithms and generalized the problem
to the case where the language in question is regular. To our knowledge, the asymptotically
most efficient algorithms for regular languages are given in \cite{Head:Weber:decision,McC:1996}, and they are
both of quadratic time complexity. Our implementation follows the algorithm in \cite{Head:Weber:decision},
as that approach makes
use of the transducer functionality algorithm. As before we enhance that algorithm to
produce a \emdef{witness of non-satisfaction} which,
given an NFA object $\aain$, it either returns the pair
(\nonen, \nonen) if $\lang(\aain)$ is
a UD code, or a pair of two different lists of $\lang(\aain)$-words such that the concatenation of the words
in each list produces the same word (if  $\lang(\aain)$ is not
a UD code).
\pssn
We now describe the algorithm in \cite{Head:Weber:decision} modified appropriately to return
a witness of non-satisfaction. Again, the heart of the algorithm is the transducer functionality
method. Let $\aut = (Q,\al,T,I,F)$ be the given NFA (with no $\ew$-transitions).
\begin{enumerate}
\item
If any of the initial states is final, then (as $\ew$ is in the language) return $(\,[\ew], [\ew,\ew]\,)$.
\item
Construct the transducer $\trt = (Q,\al,\{0,1\},T',I,I)$ in which the transition set $T'$ is defined by the
following process: If $(p,\sigma,q)\in T$ then $(p,\sigma/0,q)\in T'$ and, if in addition $q\in F$ then
also $(p,\sigma/1,i)\in T'$, for every $i\in I$. Note that the domain of the transducer
is exactly the language $\lang(\aut)^*$.
\item
Let \texttt{w, x, y = t.nonFunctionalW(a)}
\item
If any of \texttt{w, x, y} is \none then return (\nonen, \nonen)
\item
At this point, we know that $w\in\lang(\aut)$, $x$ and $y$ are different and each one is of the form
$$0^{r_1}1\cdots 0^{r_n}1.$$
Moreover, each of $x$ and $y$ corresponds to a decomposition of $w$ in terms of $\lang(\aut)$-words.
More specifically, the binary word $0^{r_1}1\cdots 0^{r_n}1$ encodes a sequence of
words $(w_i)_1^n$
such that their concatenation is equal to $w$ and each $w_i$ is the infix of $w$ that
starts at position $s_i$ and ends at position $s_i+r_i$, where $s_1=0$ and $r_i=|w_i|-1$
and $s_{i+1}=s_i+r_i+1$.
For example, if $w=ababab$ and $x=010001$, then the decomposition is $(ab,abab)$.
The algorithm returns the two lists of words corresponding to $x$ and $y$.
\end{enumerate}
\begin{EX}
The following Python interaction  produces a witness of the non-satisfaction of the UD code property by
the language $\{ab, abba, bab\}$.
\begin{verbatim}
>>> L = fl.FL(['ab', 'abba', 'bab'])
>>> a = L.toNFA()
>>> p = codes.buildUDCodeProp(a.Sigma)
>>> p.notSatisfiesW(a)
(['ab', 'bab', 'abba', 'bab'], ['abba', 'bab', 'bab', 'ab'])
\end{verbatim}
The two word lists are different, but the concatenation of the words in each list produces
the same word.
\end{EX}
\pssn
\emshort{The method} \texttt{p.maximalP(a)}\pssn
This method is based on the fundamental theorem of Schutzenberger \cite{BePeRe:2009} that a regular
language $L$ is a UD code if and only if the set of all infixes of $L^*$ is equal to $\al^*$.
Using the tools implemented in FAdo this test can be performed as follows.
\begin{verbatim}
   t = infixTransducer(a.Sigma, True)
   b = a.star()
   return (~(t.runOnNFA(b))).emptyP()
\end{verbatim}
The first statement above returns a transducer $\trt$ that, on input $w$, outputs any infix of $w$.
The second statement returns an NFA accepting $\lang(\aain)^*$, and the last statement
returns whether there is no word in the  complement of all infixes of $\lang(\aut)^*$.

\section{LaSer and Program Generation}\label{sec:laser}
The first version of LaSer \cite{DudKon:2012} was a self-contained set of C++ automaton and transducer methods as well as a set
of Python and HTML documents with the following functionality: a client uploads  a file containing an automaton in
Grail format and a file containing either a trajectory automaton, or an input altering-transducer,
and LaSer would
respond with an answer to the witness version of the satisfaction problem for input-altering transducer properties.
The new version, which we discuss here,
is based on the FAdo set of automaton and transducer methods and allows clients to request a response about
the witness versions of the satisfaction and maximality problems for input-altering transducer, error-detecting
and error-correcting properties.
\pmsn
We call the above type of functionality, where LaSer computes and returns the answer, the \emdef{online service}
of LaSer. Another  feature of the new version of LaSer, which we believe to be original in the community of software on automata and formal languages,
is the \emdef{program generation service}. This is the capability to generate a self-contained Python program that can be
downloaded on the client's machine
and executed on that machine returning thus the desired answer. This feature is useful as the execution of certain algorithms, even of
polynomial time complexity---see the error-detection satisfaction problem--, can be quite time consuming for
a server software like LaSer.
\pssn
The user interface of LaSer is very simple---see Fig.~\ref{figInterface}.
\begin{figure}[!ht]
\begin{center}
\scalebox{0.8}{\includegraphics[trim = 0.5in 5.75in 0.5in 1in, clip]{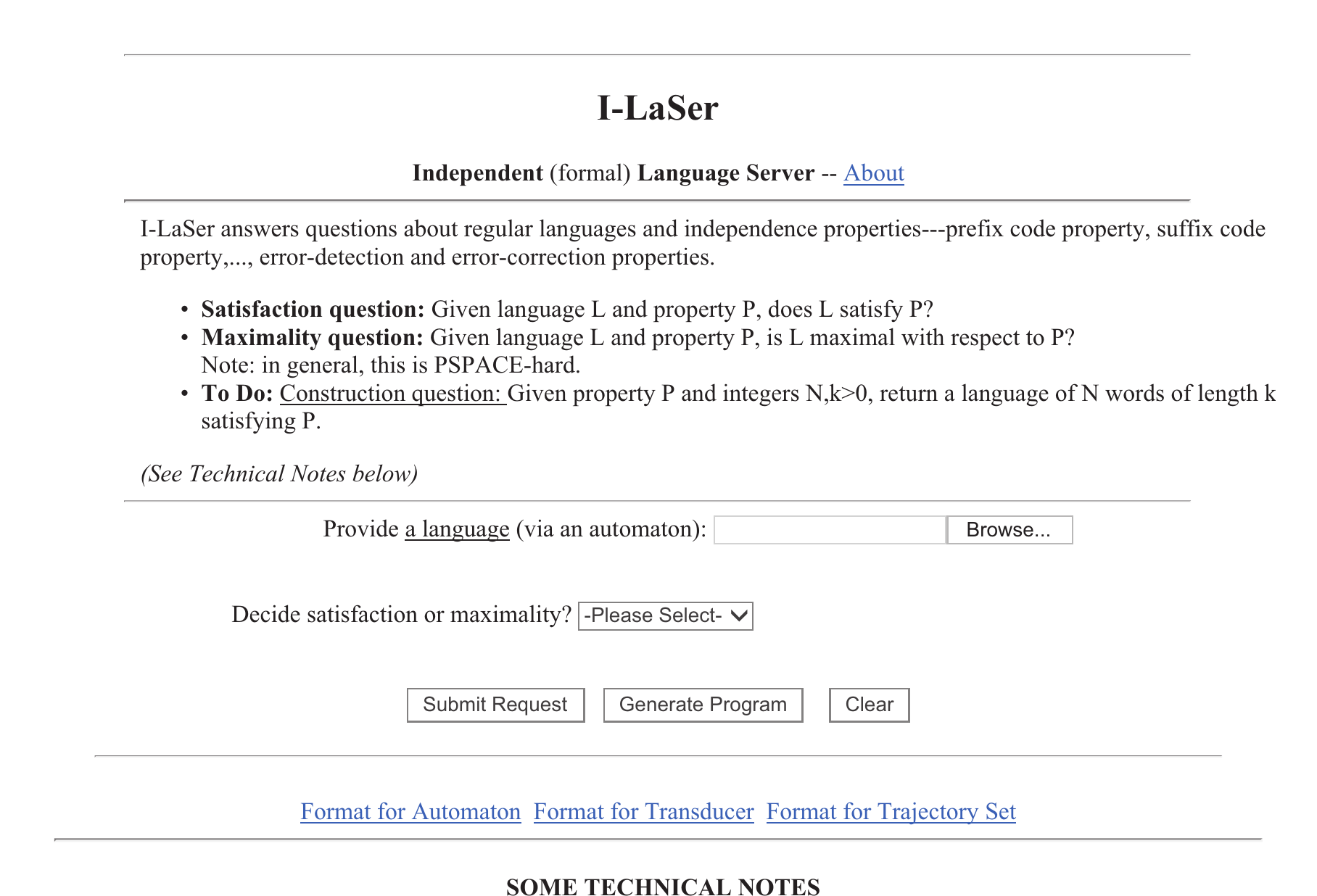}}
\parbox{0.85\textwidth}{\caption{The main user interface of LaSer}\label{figInterface}}
\end{center}
\end{figure}
The user provides the file containing the automaton whose language will be tested
for a question (satisfaction or maximality) about some property. When the decision question
is chosen, then the user is asked to enter the type of property, which can be one of
 fixed, trajectory, input-altering transducer, error-detecting, error-correcting.
 Then, the user either clicks on the button ``Submit Request'' or on the button
 ``Generate Program''.
\pmsn
Next we present some parts of the program generation module. The program to be generated will
contain code to answer one of the satisfaction or maximality questions for any of the properties
regular trajectory (TRAJECT), input-altering transducer (INALT), error-detecting (ERRDET), error-correcting (ERRCORR),
or any of the fixed properties. First a Python dictionary is prepared to enable easy reference to the
property type to be processed, and another dictionary for the question to be answered, as follows
\begin{verbatim}
buildName = {"CODE": ("buildUDCodeProperty", ["a.Sigma"], 1),
             "ERRCORR": ("buildErrorCorrectPropS", ["t"], 1),
             "ERRDET": ("buildErrorDetectPropS", ["t"], 1),
             "HYPER": ("buildHypercodeProperty", ["a.Sigma"], 1),
             "INALT": ("buildIATPropS", ["t"], 1),
             "INFIX": ("buildInfixProperty", ["a.Sigma"], 1),
             "OUTFIX": ("buildOutfixProperty", ["a.Sigma"], 1),
             "PREFIX": ("buildPrefixProperty", ["a.Sigma"], 1),
             "SUFFIX": ("buildSuffixProperty", ["a.Sigma"], 1),
             "TRAJECT": ("buildTrajPropS", ["$strexp", "$sigma"], 1)
}
\end{verbatim}
\begin{verbatim}
tests = {"MAXP": "maximalP",
         "MAXW": "notMaximalW",
         "SATP": "satisfiesP",
         "SATW": "notSatisfiesW"}.
\end{verbatim}
Next we show the actual program generation function which takes as parameters
the property name (\texttt{pname}), which must be one appearing in the dictionary \texttt{buildName},
the question to answer (\texttt{test}), the file name for the automaton
(\texttt{aname}), the possible trajectory regular expression string and alphabet (\texttt{strexp} and
\texttt{sigma}),
and the possible file name for the transducer (\texttt{tname}). The function returns a list of strings that
constitute the lines of the desired Python program---we have omitted here the initial lines that contain the
commands to import the required FAdo modules.
\begin{verbatim}
01   def program(pname, test=None, aname=None, strexp=None, sigma=None,
                 tname=None):
02       def expand(s):
03           s1 = Template(s)
04           return s1.substitute(strexp=strexp, sigma=sigma)
\end{verbatim}
\begin{verbatim}
05       l = list()
06       l.append("ax = \"%s\"\n" % base64.b64encode(aname))
07       l.append("a = readOneFromString(base64.b64decode(ax))\n")
08       if buildName[pname][2] == 1:
09           if tname:
10               l.append("tx = \"%s\"\n" % base64.b64encode(tname))
11               l.append("t = base64.b64decode(tx)\n")
12           s = "p = " + buildName[pname][0] + "("
13           for s1 in buildName[pname][1]:
14               if s1 == "$strexp":
15                   s += "t, "
16               else:
17                   s += "%s," % expand(s1)
18           s = s[:-1] + ")\n"
19           l.append(s)
20           l.append("print p.%s(a)\n" % tests[test])
21       else: ...............
\end{verbatim}
We refer to the lines above as meta-lines, as they are used to generate lines of the
desired Python program.
Meta-line~06 above generates the first line of the program, which would read the automaton file in a string
\texttt{ax} that is encoded in binary to allow
for safe transmission and reception over different operating systems. The next meta-line generates the line
that would create an NFA object from the decoded string \texttt{ax}. The if part of the \texttt{if-else} statement
is the one we use in
this paper---the else part is for other LaSer questions. If there is a transducer file then, as
in the previous case of the automaton file, LaSer generates a line that would create an SFT object from the encoded file.
Then, meta-lines~12--18 generate a line that would create the property \texttt{p} to be processed,
using the appropriate build-property function. Finally, meta-line~20 generates the line that would
print the result of invoking the desired method  \texttt{test} of the property \texttt{p}.


\section{Concluding Remarks}\label{sec:conclude}
We have presented a simple to use set of methods and functions that allow one
to define many natural code properties and obtain answers to the satisfaction and
maximality problems with respect to these properties. This capability relies on our
implementation of basic transducer methods, including our witness version of the non-functionality
transducer method, in the  FAdo set of Python packages.
We have also produced a new version of LaSer that allows clients to inquire about
error-detecting and -correcting properties, as well as to generate programs that
can be executed and provide answers at the client site.
\pmsn
There are a few important directions for future research. First, the existing implementation of
transducer objects is not always efficient when it comes to describing code properties.
For example, the transducer defined in Example~\ref{exSub1} consists of 6 transitions.
In general, if the alphabet is of size $s$, then that transducer would require $s+s(s-1)+s=s^2+s$
transitions. However, a symbolic notation for transitions, say of the form
\begin{verbatim}
  0 @s @s 0
  0 @s @s' 1
  1 @s @s 1
\end{verbatim}
is more compact and can possibly be used by modifying the appropriate transducer methods.
In this hypothetical notation, \verb1@s1 represents any alphabet symbol and \verb1@s'1 represents any
alphabet symbol other than the previous one. Of course the syntax of
such transducer descriptions has to be defined carefully so as to be as expressive as possible and
at the same time efficiently utilizable in transducer methods.
We note that a similar idea for symbolic transducers is already investigated in \cite{Vea:2013}.
\pssi
Formal methods for defining code properties need to be better understood or new ones need to be developed
with the aim of ultimately implementing these properties and answering  efficiently the satisfaction
problem. Moreover, these methods  should be capable of allowing to express properties that cannot
be expressed in the transducer methods considered here. In particular, all transducer properties
in this work are independences with parameter $n=3$, so they do not include, for instance, the
comma-free code property. A language $L$ is a comma-free code~\cite{Shyr:book} if
\[
LL\cap \al^+ L\al^+=\emptyset.
\]
The formal method of~\cite{Jurg:1999} is quite expressive, using  a certain type
of first order formulae to describe properties. It could perhaps be  further
worked out in a way that some of these formulae can be mapped to finite-state machine
objects that are, or can be, implemented in available formal language packages like FAdo.
We also note that if the defining method is too expressive then
even the satisfaction problem could become undecidable---see for example the method of
multiple sets of trajectories in \cite{DomSal:2006}.
\pmsn
We consider the construction problem to be the Holy Grail of coding theory. We believe
that as long as the satisfaction problem is efficiently decidable one can use an algorithmic
approach to address the problem to some extent. For example, we already have
implemented a randomized algorithm that starts with
an initial small language $L$ (in fact a singleton one) and then performs  a loop in which
it does the following task.
It  randomly  picks words $w$ of length $\ell$ that are not in $L$ until
$L\cup w$ satisfies the property,
in which case $L$ becomes $L\cup w$, or no such $w$ is found after a MAX value
of trials. 
Although this is a simple approach, we feel that it can lead to further developments with respect to
the construction problem.

\bibliographystyle{plain}

\end{document}